\documentclass[runningheads]{llncs}
\usepackage[T1]{fontenc}
\usepackage{graphicx}
\usepackage{calc}
\usepackage{amsmath}
\usepackage{amssymb}
\usepackage{color}
\usepackage[capitalise,noabbrev]{cleveref}
\usepackage[disable]{todonotes}
\usepackage{enumitem}
\usepackage{subcaption}
\usepackage{soul}
\usepackage{tikz}
\usetikzlibrary{automata, arrows.meta, 
positioning, decorations.pathreplacing, calligraphy, shapes.geometric}

\newcommand{\cA}{\mathcal{A}}
\newcommand{\cB}{\mathcal{B}}
\newcommand{\cD}{\mathcal{D}}
\newcommand{\cM}{\mathcal{M}}
\newcommand{\cN}{\mathcal{N}}
\newcommand{\gap}[3]{\mathsf{gap}(#1,#2,#3)}
\newcommand{\bbN}{\mathbb{N}}
\newcommand{\val}{\texttt{val}}
\newcommand{\fval}{\texttt{fval}}

\newcommand{\disc}[1]{\mathcal{\lambda}^{-#1}}
\newcommand{\undisc}[1]{\mathcal{\lambda}^{#1}}
\newcommand{\lam}{\lambda}
\newcommand{\lami}{\lambda^{-1}}

\begin{document}
\title{Determinization of Integral Discounted-Sum Automata is Decidable\thanks{This research was supported by the ISRAEL SCIENCE FOUNDATION (grant No. 989/22)}}
\author{Shaull Almagor\orcidID{0000-0001-9021-1175} \and
Neta Dafni}
\authorrunning{S. Almagor and N. Dafni}
\institute{Technion, Israel \\
\email{shaull@technion.ac.il} \\ 
\email{netad@campus.technion.ac.il}}
\maketitle              % typeset the header of the contribution

\begin{abstract}
Nondeterministic Discounted-Sum Automata (NDAs) are nondeterministic finite automata equipped with a discounting factor $\lambda>1$, and whose transitions are labelled by weights. The value of a run of an NDA is the discounted sum of the edge weights, where the $i$-th weight is divided by $\lambda^{i}$. NDAs are a useful tool for modelling systems where the values of future events are less influential than immediate ones.

While several problems are undecidable or open for NDA, their deterministic fragment (DDA) admits more tractable algorithms. Therefore, determinization of NDAs (i.e., deciding if an NDA has a functionally-equivalent DDA) is desirable.

Previous works establish that when $\lambda\in \bbN$, then every \emph{complete} NDA, namely an NDA whose states are all accepting and its transition function is complete, is determinizable. This, however, no longer holds when the completeness assumption is dropped.

We show that the problem of whether an NDA has an equivalent DDA is decidable when $\lambda\in \bbN$ (in particular, it is in $\mathsf{EXPSPACE}$ and is $\mathsf{PSPACE-hard}$).

\end{abstract}
%
%
%

%\shtodo{Note that I changed the $\cA[\to]$ notation to be in subscript. In the future, make this kind of thing a macro, so that it's easy to change.}
\section{Introduction}
Traditional methods of modelling systems rely on Boolean automata, where every word is assigned a Boolean value (i.e., accepted or rejected). This setting is often generalized into a richer, quantitative one, where every word is assigned a numerical value, and thus the Boolean concept of a language, i.e., a set of words, is lifted to a more general function, namely a function from words to values.

%There are numerous automata models for capturing such functions
%A large class stems from semirings -- algebraic structures with two operations, where one operation (denoted $\otimes$) is used to track the weights along runs, and a second operation $\oplus$ resolves the nondeterminism between runs. A prominent example of this model is \emph{tropical weighted automata}, where the semiring is the $\min,+$ semiring.

%Another class of models is \emph{quantitative automata} -- nondeterministic finite automata equipped with a weight function on the transitions, whose every run is assigned a value as a function of its transition weights; that function can, to give a few examples, use summation, product or average over the weights. The value of every word is taken as minimum or maximum over all accepted runs. In this sense, Quantitative automata can capture tropical automata, but also other models.

A particular instance of quantitative automata is that of \emph{discounted-sum automata}. There, the weight function sums the weights along the run, but discounts the future. Discounting as a general notion is a well studied concept in game theory and various social choice models~\cite{broome1994discounting}. Computational models with discounting, such as 
discounted-payoff games \cite{ZWICK1996343,Andersson2006AnIA,10.1007/3-540-45061-0_79}, discounted-sum Markov Decision Processes \cite{4276554,10.1145/1721837.1721849,10.1007/978-3-642-45221-5_17} and discounted-sum automata \cite{DROSTE2006199,10.1007/978-3-642-03409-1_2,5230579,10.1007/978-3-540-87531-4_28}, 
 are therefore useful to model settings where the far future has less influence than the immediate future.

In this work we focus on non-deterministic discounted-sum automata (NDAs). An NDA is a quantitative automaton equipped with a \emph{discounting factor} $\lambda>1$. The value of a run is the discounted sum of the transitions along the run, where the value of transition $i$ is divided by $\lambda^i$.  The value of a word is then the value of the minimal accepting run on it. We also allow \emph{final weights} that are added to the run at its end (with appropriate discounting).
%NDAs differ from WFAs in that when summing the weight of a run, the weight of the $i$'th transition is multiplied by $\disc{i}$. A word is then assigned the minimum value over all its accepting runs.

Unlike Boolean automata, NDAs are strictly more expressive than their deterministic counterpart (DDAs)~\cite{10.1007/978-3-540-87531-4_28}. In particular, certain decision problems for NDAs are undecidable, but become decidable for DDAs~\cite{boker2023comparison}.
There is, however, a subclass of NDAs that always admit an equivalent DDA: the \emph{complete integral NDAs}~\cite{boker_et_al:LIPIcs:2011:3224}. An automaton is \emph{complete} if its transition function is total and all its states are accepting with final weight $0$. This means that runs never ``die'', and that all runs are accepting. An NDA is \emph{integral} if its discounting factor $\lambda$ is an integer. 
It is further shown in~\cite{boker_et_al:LIPIcs:2011:3224} that if the completeness requirement is removed then for every discounting factor there is an integral NDA that is not determinizable. 

%Nondeterministic models are notoriously complicated to reason about, due to their many possible behaviours. In the quantitative setting, nondeterminism may even cause some problems to become undecidable, due to the possibly-infinite range of values for different runs~\cite{almagor2022s}. Unfortunately, in contrast with the Boolean setting, not all nondeterministic quantitative automata admit equivalent deterministic ones. 

%Thus, a central decision problem for quantitative models is that of \emph{determinization}: given a nondeterministic quantitative automaton, does it have an equivalent deterministic one? (where by equivalent we mean that it induces the same function).
%Tropical WFAs, in general, cannot be determinized \cite{10.1007/978-3-540-87531-4_28,10.1007/978-3-642-24372-1_2} and the decidability of their determinization problem is a major open question.
%Certain restrictions on WFAs related to their ambiguity (a measure for the ``degree'' of nondeterminism) lead to decidable determinization \cite{mohri-1997-finite,kirsten_2008}.  

The existence of NDAs that are not determinizable implies that the determinization problem is not trivial. However, its decidability and complexity have not been studied. In this work, we show that determinization of integral NDAs is decidable. Specifically, we show that determinization is in $\mathsf{EXPSPACE}$ and is $\mathsf{PSPACE-hard}$.

%In this work we consider NDAs and further focus on integral ones ($\lambda$ is an integer). As mentioned above, the problem remains nontrivial under this model.
\begin{example}
\label{xmp:nondet}
We demonstrate the determinization problem, as well as some intricacies involved in its analysis.
Consider the NDA in \cref{subfig:NDA_det_example}. Intuitively, the NDA either reads only $a$'s, or reads a word of the form $a^*b$. However, it guesses in $q_0$ whether it is going to read many $a$'s, in which case it may be worthwhile incurring weight $3$ to $q_2$ in order to read the remaining $a$'s at cost $0$.
\setlength{\intextsep}{7pt}
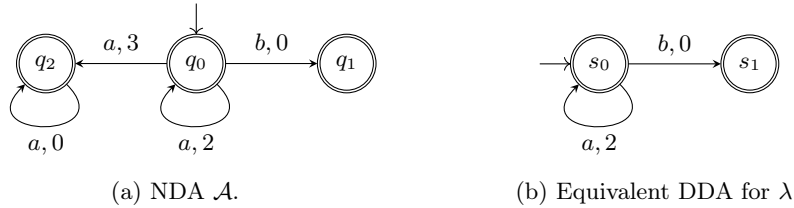
\begin{figure}
\centering
\begin{subfigure}{0.45\textwidth}     
\begin{tikzpicture}[auto,node distance=2.8cm,scale=1]
        \tikzset{every state/.style={minimum size=20pt, inner sep=4}};
        \node (q0) [initial above,accepting, state, initial text = {}] at (0,0) {$q_0$};
        \node (q1) [state,accepting] at (2,0) {$q_1$};
        \node (q2) [state,accepting] at (-2,0) {$q_2$};
        \path [-stealth]
        (q0) edge node [above] {$b,0$} (q1)
        (q0) edge [loop below, distance=30,out=-45,in=225] node {$a,2$} (q0)
        (q0) edge node [above] {$a,3$} (q2)
        (q2) edge [loop below,distance=30,out=-45,in=225] node {$a,0$} (q2);
    \end{tikzpicture}
    \caption{NDA $\cA$.}
    \label{subfig:NDA_det_example}
    \end{subfigure}
    \hfill
    \centering
    \begin{subfigure}{0.45\textwidth}     
    \qquad
\begin{tikzpicture}[auto,node distance=2.8cm,scale=1]
        \tikzset{every state/.style={minimum size=20pt, inner sep=4}};
        \node (q0) [initial,accepting, state, initial text = {}] at (0,0) {$s_0$};
        \node (q1) [state,accepting] at (2,0) {$s_1$};
        \path [-stealth]
        (q0) edge node [above] {$b,0$} (q1)
        (q0) edge [loop below, distance=30,out=-45,in=225] node {$a,2$} (q0);
    \end{tikzpicture}
    \caption{Equivalent DDA for $\lambda=3$.}
    \label{subfig:DDA_det_example}
    \end{subfigure}
\caption{The NDA $\cA$ on the left is determinizable with $\lambda=3$, with an equivalent DDA depicted on the right. However, $\cA$ is not determinizable with $\lambda=2$.}
\label{fig: NDA intro}
\end{figure}

We now ask if this NDA has a deterministic equivalent. As it turns out, this is dependent on the discounting factor. Indeed, consider the discounting factor $\lambda=3$, then when reading the word $a^k$, the run that remains in $q_0$ has weight $\sum_{i=0}^k 2\cdot 3^{-i}=3-3^{-k}<3$, whereas a run that moves to $q_2$ at step $j\ge 0$ has weight $\sum_{i=0}^{j-1}2\cdot 3^{-i}+3\cdot 3^{-j}=3-3^{-j}+3^{-j}=3$. Thus, it is always beneficial to remain in $q_0$. In this case, we do have a deterministic equivalent, depicted in \cref{subfig:DDA_det_example}. We remark that the fact that this deterministic equivalent is obtained by removing transitions is not a standard behaviour, and typically determinization involves a blowup.

Next, consider the discounting factor $\lambda=2$. Similar analysis shows that for the word $a^k$, the weight of the run that stays in $q_0$ is $\sum_{i=0}^k 2\cdot 2^{-i}=4-2\cdot 2^{-k}$, whereas leaving to state $q_2$ at step $0$ yields cost $3$, so the latter is preferable for large $k$. Intuitively, this means that nondeterminism is necessary in this setting, since the NDA does not ``know'' whether $b$ will be seen. Indeed, for $\lambda=2$ this NDA is not determinizable.

Observe that in the case of $\lambda=2$, the two ``extreme'' runs on $a^k$, namely the one that stays in $q_0$ and the one that immediately leaves to $q_2$, create a ``gap'' between their values that tends toward $1$ as $k$ increases. Keeping in mind that for large $k$ the transition value is multiplied by $2^{-k}$, intuitively this gap becomes huge. As we show in this work, this concept of gaps exactly characterizes whether an NDA can be determinized.
\hfill\qed
\end{example}

We remark that for non-integral NDAs, many problems, including the determinization problem, are open due to number-theoretic difficulties~\cite{boker2015target}. Therefore, it is unlikely that progress is made there, pending breakthroughs in number theory. 

\paragraph{Related Work}
Discounted-sum automata have been studied in various contexts. Specifically, certain algorithmic problems for them are still open, and are closely related to longstanding open problems~\cite{boker2015target}. In addition, they are not closed under standard Boolean operations~\cite{boker2014exact} (which is often the case in quantitative models, due to the ``minimum'' semantics which conflicts with notions of conjunction).

Recently, discounted sum automata were also studied in the context of two-player games~\cite{cadilhac2019impatient}. Of particular interest are ``regret-minimizing strategies'', where the concept of regret minimization is closely related to determinization of automata~\cite{filiot2017delay}.

An extension of discounted-sum automata to multiple discounting factors (NMDAs) was studied in~\cite{boker_et_al:LIPIcs:2021:13446,boker2023comparison}, where NMDAs are NDAs where every transition is allowed a different discounting factor. 
NMDAs are generally non-determinizable, but imposing certain restrictions on the choice of discounting factors can ensure determinizability \cite{boker_et_al:LIPIcs:2021:13446}. We remark that the study of NMDAs is still only with respect to complete automata.

Determinization of other quantitative models has also received some attention in recent years. A major open problem is the decidability of determinization for weighted automata over the tropical semiring (for some subclasses it is known to be decidable~\cite{mohri-1997-finite,kirsten_2008}). Interestingly, a tropical weighted automaton can be seen as the ``limit'' of NDAs where $\lambda\to 1$. This, however, does not seem to help in resolving the decidability of the former.

In \cite{2f4c7b43623c4e64883e3032b3b05950}, the determinization problem for one-counter nets (OCNs) is studied. OCNs are automata equipped with a \emph{counter} that cannot decrease below zero. They can be thought of as pushdown automata with a singleton stack alphabet. Most notions of determinizability introduced in \cite{2f4c7b43623c4e64883e3032b3b05950} are undecidable, with one case being open (and seemingly related to the setting of weighted automata).

Due to space constraints, some proofs appear in the appendix.

\section{Preliminaries}
A \emph{nondeterministic integral discounted-sum automaton} (NDA) is a tuple 
$\cA=(\Sigma,Q,Q_0,\alpha,\delta,\val,\fval,\lambda)$, where $\Sigma$ is a finite alphabet, $Q$ is a finite set of states, $Q_0\subseteq Q$ is a set of initial states, $\alpha\subseteq Q$ is a set of accepting states, $\delta\subseteq Q\times\Sigma\times Q$ is a transition relation, $\val:\delta\to\mathbb{Z}$ is a \emph{weight function} that assigns to each transition $(p,\sigma,q)\in \delta$ a \emph{weight} $\val(p,\sigma,q)\in \mathbb{Z}$, $\fval:\alpha\to\mathbb{Z}$ is a \emph{final weight function} that assigns a \emph{final weight}\footnote{In some works, the weights are assumed to be rational. For determinizability we can assume all weights are integers, since we can always multiply every weight by a common denominator.} to every accepting state, and $1<\lambda\in\bbN$ is an integer \emph{discounting factor}.

%We remark that while traditionally the weights are allowed be rational, restricting them to the integers essentially does not change the model, as integer weights can be obtained from rational weights by multiplying all weights by a constant integer, which results in an NDA that realizes the same function as the original up to a constant. In particular, the determinizability or non-determinizability of the NDA does not change.

The existence of a transition $(p,\sigma,q)\in\delta$ means that when $\cA$ is in state $p$ and reads the letter $\sigma$ it can move to state $q$. If there exists $q$ such that $(p,\sigma,q)\in \delta$, we say that $p$ has a \emph{$\sigma$-transition}. If $p$ does not have a $\sigma$-transition, that means that when in state $p$ and reading the letter $\sigma$, $\cA$'s run cannot continue.

Consider a word $w=w_1 \cdots w_{n}\in\Sigma^*$. A \emph{run} of $\cA$ on $w$ is a sequence of states $\rho=\rho_0,\rho_1, \ldots ,\rho_{n}$ such that $\rho_0\in Q_0$ and $(\rho_{i-1},w_i,\rho_i)\in\delta$ for every $1\leq i\leq n$. The run is \emph{accepting} if $\rho_{n}\in \alpha$.
The \emph{weight} of $\rho$ is the discounted sum $\val(\rho)=\Sigma_{i=0}^{n-1} \disc{i}\val(\rho_i,w_{i+1},\rho_{i+1})$. 

The \emph{value} of $w$ by $\cA$, denoted $\cA^*(w)$, is $\min\{\val(\rho)+\disc{n}\fval(\rho_n)\mid\rho=\rho_0,\ldots,\rho_n \text{ is an accepting run on \ensuremath{w}}\}$, that is, the minimal weight of a run on $w$ including final weights, or $\infty$ if no such run exists. Two NDAs $\cA, \cB$ are \emph{equivalent} if $\cA^*(w)=\cB^*(w)$ for every $w\in \Sigma^*$. 

We say that $\cA$ is \emph{deterministic} (DDA, for short) if $|Q_0|=1$ and $\{q\in Q|(p,\sigma,q)\in\delta\}\leq 1$ for every $p\in Q,\sigma\in\Sigma$. Note that if $\cA$ is deterministic then for every word there is at most one run starting in each state. For a DDA we define the partial function $\delta^*:Q\times\Sigma^*\hookrightarrow Q$ such that $\delta^*(q,w)$ is the final state in the run on $w$ starting in $q$, if such a run exists. We say that an NDA $\cA$ is \emph{determinizable} if it has an equivalent DDA. 

It will also be useful to consider non-accepting runs and runs that start and end in specific states of $\cA$. For sets of states $P,P'\subseteq Q$ we define $\cA_{[P\to P']}(w)$ to be the weight of a minimal run of $\cA$ on $w$ from some state in $P$ to some state in $P'$. Similarly, $\cA_{[P\to_f P']}(w)$ is the minimal weight of an accepting run including final weights. When $P$ or $P'$ is a singleton $\{p\}$ we omit the parenthesis. When $P=Q_0$ and $P'=Q$ (or $\alpha$, for the setting of including final weights) we omit the sets and write e.g., $\cA(w)$ instead of $\cA_{[Q_0\to Q]}(w)$, and $\cA^*(w)$ instead of $\cA_{[Q_0\to_f \alpha]}(w)$. Under these notations, if a run does not exist, the assigned value is $\infty$.
For the remainder of the paper, fix an integral NDA $\cA$. We assume that $\cA$ is \emph{trim}, i.e., for every $q\in Q$ there exists $p\in\alpha$ that is reachable from $q$ in the underlying graph of $\cA$. Since it is easy to compute a trim equivalent for a given NDA, this assumption is valid. It follows that for every $q$ there exists a word $y$ such that $\cA_{[q\to \alpha]}(x)<\infty$.

% for every $q\in Q$ there exist $x,y\in\Sigma^*$ such that $\cA_{[Q_0\to q]}(x)<\infty$ and $\cA_{[q\to_f\alpha]}(y)<\infty$. Since it is easy to compute a trim equivalent for a given NDA, this assumption is valid.
% \shtodo{It would be better to define trim as ``there is a state in $\alpha$ reachable from $q$ in the underlying graph'', and then explain that this condition is entailed. Otherwise it's less clear that it's easily computable.}

\section{Gaps and Separation of Runs}
\label{sec:gaps}
In this section we lay down the basic definitions we use throughout the paper, concerning the ways several runs on the same word accumulate different weights.

Denote by $m_{\cA}$ the maximal absolute value of a weight of a transition or a final weight in $\cA$. Recall that the geometric sum  (for $\lambda>1$) satisfies $\sum_{i=0}^\infty \disc{i}=\frac{\lambda}{\lambda-1}$. Therefore, 
$\frac{\lambda}{\lambda-1}m_{\cA}$ is an upper bound on $|\val(\rho)|$ for any run $\rho$. Indeed, we have
%\begin{align*}
%|\val(\rho_0, \ldots ,\rho_n)|= & |\Sigma_{i=0}^{n-1} \disc{i}\val(\rho_i,w_{i+1},\rho_{i+1})|\\
%\leq & \Sigma_{i=0}^{n-1} \disc{i}m_\cA\\
%< & \frac{\lambda}{\lambda-1}m_\cA
%\end{align*}
\[
|\val(\rho_0, \ldots ,\rho_n)|=  |\Sigma_{i=0}^{n-1} \disc{i}\val(\rho_i,w_{i+1},\rho_{i+1})|
\leq  \Sigma_{i=0}^{n-1} \disc{i}m_\cA
<  \frac{\lambda}{\lambda-1}m_\cA
\]
Clearly, the same bound holds when including final weights.

Let $\cM=2\frac{\lambda}{\lambda-1}m_{\cA}$, then for every two runs $\rho^1,\rho^2$ we have $|\val(\rho^1)-\val(\rho^2)|<\frac{\lambda}{\lambda-1}m_\cA-(-\frac{\lambda}{\lambda-1}m_\cA)= \cM$. The constant $\cM$ is central in our study of gaps between runs.

Consider a word $w\in \Sigma^*$. The run attaining the minimal value $\cA^*(w)$ might not be minimal while reading prefixes of $w$. The \emph{gap} between the value of an eventually-minimal run and minimal runs on prefixes of $w$ is central to characterizing determinizability of NDAs \cite{boker_et_al:LIPIcs:2011:3224}. This gap is captured by the following definitions. We consider two versions, one where the eventual minimality refers to accepting runs, and one where it refers to all runs.

\begin{definition}[Recoverable gap]
\label{def: recoverable gap}
Consider words $w,z\in \Sigma^*$ and states $q_u,q_l\in Q$. the tuple $(w,q_u,q_l)$ is called a \emph{terminal recoverable gap (TRG) with respect to $z$}, or simply a \emph{terminal recoverable gap}, if the following hold:
\begin{enumerate}
    \item $\cA_{[Q_0\to q_l]}(w)\leq\cA_{[Q_0\to q_u]}(w)$, and
    \item $\cA_{[Q_0\to q_u]}(w)+\disc{|w|}\cA_{[q_u\to_f \alpha]}(z)=\cA^*(wz)<\infty$. 
\end{enumerate}

A \emph{global recoverable gap (GRG)} is defined similarly, except the second requirement is replaced with
\[
    \cA_{[Q_0\to q_u]}(w)+\disc{|w|}\cA_{[q_u\to \alpha]}(z)=\cA(wz)<\infty
\]
We use \emph{recoverable gap} to refer to a gap of either type.
\end{definition}

Intuitively, in a recoverable gap $(w,q_u,q_l)$ there are runs $\rho_1$ and $\rho_2$ of $\cA$ on $w$ that end in $q_u$ and $q_l$, respectively, where $\rho_1$ attains a higher value than $\rho_2$, but there is a suffix $z$ that ``recovers'' this gap: when reading $z$ from $q_u$ starting with weight $\val(\rho_1)$, the resulting minimal run attains the minimal value of a run of $\cA$ on $wz$. This is depicted in \cref{fig: recoverable gap}.

\begin{figure}[ht]
\begin{center}

\begin{tikzpicture}[scale=0.15]

\draw [black] (-4.47,0) -- (48.95,0);
\draw [black] (-4.47,-1) -- (-4.47,1);
\draw [black] (35,-1) -- (35,1);
\draw [black] (48.95,-1) -- (48.95,1);
\draw (15.265,-2) node {$w$};
\draw (41.975,-2) node {$z$};

\draw [black,fill=black] (-4.47,2.67) circle (0.4);

\draw [black] (-4.47,2.67) arc (150:100:20);
\draw [black] (9.38,12.37) arc (280:300:60);
\draw [black] (28.96,19.5) arc (120:85:10);
\draw [dashed,black] (35,20.78) arc (85:70:10);
\draw [dashed,black] (37.38,20.24) arc (70:20:15);
\draw [dashed,black] (46.35,11.27) arc (200:220:15);
\draw [black,fill=black] (35,20.78) circle (0.4);
\draw (35,19) node {$q_u$};
\draw [black,fill=black] (48.95,6.76) circle (0.4);
\draw (15.265,15.5) node {$\rho^u_1$};
\draw (41.975,19.5) node {$\rho^u_2$};

\draw [black] (-4.47,2.67) arc (270:290:65);
\draw [black] (17.76,6.59) arc (110:70:10);
\draw [black] (24.6,6.59) arc (250:280:20);
\draw [dashed,black] (35,5.75) arc (280:300:20);
\draw [dashed,black] (41.44,8.06) arc (300:350:10);
\draw [dashed,black] (46.29,14.98) arc (170:101:4);
\draw [black,fill=black] (35,5.75) circle (0.4);
\draw (35,7.5) node {$q_l$};
\draw [black,fill=black] (48.95,18.12) circle (0.4);
\draw (15.265,4.2) node {$\rho^l_1$};
\draw (41.975,6.8) node {$\rho^l_2$};

\end{tikzpicture}
\caption{The run $\rho^l_1$, ending in state $q_l$, is the minimal run of $\cA$ on $w$. The higher run $\rho^u_1$ is the minimal run on $w$ that ends in state $q_u$, thus creating a gap between $q_u$ and $q_l$. However, the concatenation $\rho^u_1 \cdot \rho^u_2$ is the minimal run on the concatenated word $wz$, while the concatenation $\rho^l_1 \cdot\rho^l_2$, where $\rho^l_2$ is the minimal run on $z$ starting in $q_l$, is not smaller. Therefore, the gap is recoverable. Note that here the final weights are zero.}
\label{fig: recoverable gap}
\end{center}
\end{figure}
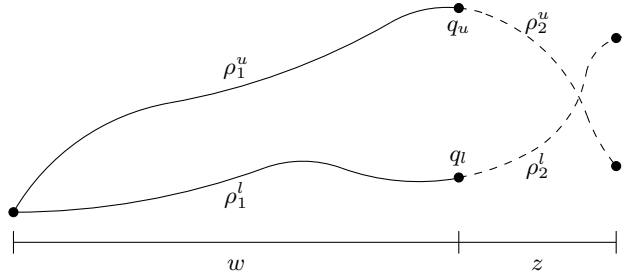

For a recoverable gap $(w,q_u,q_l)$ we define 
\[
    \gap{w}{q_u}{q_l}=\undisc{|w|}(\cA_{[Q_0\to q_u]}(w)-\cA_{[Q_0\to q_l]}(w))
\]
The normalizing factor $\undisc{|w|}$ eliminates the effect of the length of $w$ on the gap, allowing us to study gaps independently of the length of their corresponding words.

We say that $\cA$ has \emph{finitely/infinitely many TRGs/GRGs/recoverable gaps} if the set $\{ \gap{w}{q_u}{q_l}\mid (w,q_u,q_l) \text{ is a TRG/GRG/recoverable gap}\} $
is finite/infinite, respectively. Note that since $\cA$ is integral, $\undisc{|w|}(\cA_{[Q_0\to q_u]}(w)-\cA_{[Q_0\to q_l]}(w))$ is always an integer and so the existence of infinitely many TRGs/GRGs/recoverable gaps is equivalent to the existence of unboundedly large TRGs/GRGs/recoverable gaps, respectively.

While gaps refer to two distinct runs, we sometimes need a more global view of gaps.  To this end, we lift the definition to all the reachable states, as follows.
\begin{definition}[$n$-separation]
\label{def: separation property} 
For a word $w$ and $n\in\mathbb{N}$, we say that $w$ has the \emph{$n$-separation property} if there exists a partition of $Q$ into two non-empty sets of states $U,L$ such that the following holds:
\begin{enumerate}
\item For every $q_{u}\in U$ and $q_{l}\in L$, $\undisc{|w|}(\cA_{[Q_0\to q_u]}(w)-\cA_{[Q_0\to q_l]}(w))>n$. 
\item There exist $q_u\in U$ and $z\in \Sigma^*$ such that for every $q_l\in L$,
$(w,q_u,q_l)$ is a TRG with respect to $z$.
\end{enumerate}
We sometimes explicitly specify that $w$ \emph{has the $n$-separation property with respect to $(U,L,q_{u})$}, or \emph{with respect to $(U,L,q_{u},z)$}. If there exists $w$ with the $n$-separation property, we say that $\cA$ has the $n$-separation property.
\end{definition}
See \cref{fig: n-separation} for a depiction of $n$-separation. Note that \cref{def: separation property} refers only to TRGs and we do not need an analogous version for GRGs. That is because the determinizability of $\cA$ is characterized by the number of TRGs rather than GRGs.

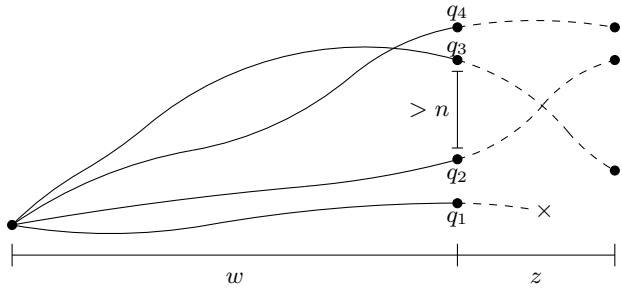
\begin{figure}[ht]
\begin{center}
\begin{tikzpicture}[scale=0.15]

\draw [black] (-4.47,0) -- (48.95,0);
\draw [black] (-4.47,-1) -- (-4.47,1);
\draw [black] (35,-1) -- (35,1);
\draw [black] (48.95,-1) -- (48.95,1);
\draw (15.265,-2) node {$w$};
\draw (41.975,-2) node {$z$};

\draw [black,fill=black] (-4.47,2.67) circle (0.4);

\draw [black] (-4.47,2.67) arc (260:280:50);
\draw [black] (12.89,2.67) arc (100:90:128);
\draw [black,fill=black] (35,4.6) circle (0.4);
\draw (35,3.2) node {$q_1$};
\draw [dashed,black] (35,4.6) arc (90:80:40);
\draw (42.7,3.9) node {$\times$};

\draw [black] (-4.47,2.67) arc (100:95:300);
\draw [black] (21.48,6.09) arc (275:285:79);
\draw [black,fill=black] (35,8.5) circle (0.4);
\draw [dashed,black] (35,8.5) arc (285:320:15);
\draw (35,7) node {$q_2$};
\draw [dashed,black] (42.6,13.35) arc (140:102:12);
\draw [black,fill=black] (48.95,17.3) circle (0.4);

\draw [black] (-4.47,2.67) arc (125:100:40);
\draw [black] (11.53,9.3) arc (280:310:30);
\draw [black] (25.6,15.86) arc (130:100:20);
\draw [black,fill=black] (35,20.2) circle (0.4);
\draw (35,21.5) node {$q_4$};
\draw [dashed,black] (35,20.2) arc (100:80:40);
\draw [black,fill=black] (48.95,20.2) circle (0.4);

\draw [black] (-4.47,2.67) arc (135:120:30);
\draw [black] (1.74,7.44) arc (300:310:40);
\draw [black] (7.45,11.44) arc (130:74:30);
\draw [black,fill=black] (35,17.3) circle (0.4);
\draw (35,18.3) node {$q_3$};
\draw [dashed,black] (35,17.3) arc (74:40:20);
\draw [dashed,black] (44.8,10.93) arc (220:240:15);
\draw [black,fill=black] (48.95,7.5) circle (0.4);

\draw [black] (35,9.5) -- (35,16.3);
\draw [black] (34.5,9.5) -- (35.5,9.5);
\draw [black] (34.5,16.3) -- (35.5,16.3);
\draw (32.5,12.9) node {$>n$};

\end{tikzpicture}
\caption{Depicted are minimal runs of an NDA on a word $w$ that end in each of four states, $q_1,q_2,q_3,q_4$, and minimal runs on $z$ starting in each of them. The run from $q_1$ (lowest) ``gets stuck'', i.e., such a run from $q_1$ on $z$ does not exist. The states are partitioned into two sets $L=\{q_1,q_2\}$ and $U=\{q_3,q_4\}$, with a gap larger than $n$ between them after reading $w$; additionally, one of the upper runs then becomes minimal after reading $z$, since each of the lower runs either ends higher or ``gets stuck''. This means that the word $w$ has the $n$-separation property with respect to $(U,L,q_3,z)$. %Note that here the final weights are zero.
}
\label{fig: n-separation}
\end{center}
\end{figure}

\section{Determinizability of Integral NDAs is Decidable -- Proof Overview}
Recall that our goal is to show the decidability of the determinization problem.

As showed in~\cite{boker_et_al:LIPIcs:2011:3224}, determinizability is closely related to recoverable gaps. More precisely, a DDA $\cD$ that ``attempts'' to be equivalent to $\cA$ must keep track of all the relevant runs of $\cA$. If two runs end in the same state, it is clearly enough to track only the minimal one. However, this may still require keeping track of runs that attain unboundedly high values (when normalized). Therefore, in order for $\cD$ to be finite, it must discard information on runs that get too high. The main issue is whether we can give a bound above which runs are no longer relevant.

For \emph{complete} integral NDAs, there are always finitely many recoverable gaps, and this is used to show that complete NDAs are always determinizable~\cite{boker_et_al:LIPIcs:2011:3224}. 
For a general integral NDA $\cA$, we similarly show in \cref{sec:fin gaps implies det} that if there are only finitely many recoverable gaps, then $\cA$ is determinizable.

There are now two main challenges: First, to show that if $\cA$ has infinitely many recoverable gaps, then it is not determinizable, and second, that it is decidable whether $\cA$ has finitely many recoverable gaps.

%Both challenges amount to reasoning about recoverable gaps. We therefore 
We start by showing in \cref{sec: inf GRG to inf TRG} that $\cA$ has infinitely many recoverable gaps if and only if it has infinitely many TRGs.
In ~\cref{sec:large gap to inf gaps}, we show that we can compute a bound $\cN$ such that $\cA$ has infinitely many TRGs if and only if it has a TRG larger than $\cN$.
Next, in \cref{sec:gap to separation}, we show that the existence of a TRG than $\cN$ is also equivalent to some word having the $\cN$-separation property. 

We then turn to exhibit a small-model property on witnesses for $\cN$-separation. Specifically, we show in~\cref{sec:small model} that if there exist $w,z$ such that $w$ has the $\cN$-separation property with respect to $(U,L,q_{u},z)$, then we can bound the length of the shortest $w,z$.

Using the above, we obtain the decidability of whether $\cA$ has infinitely many recoverable gaps. In addition, we use these results to prove (in \cref{lem: det -> finite TRGs}) that if $\cA$ has infinitely many TRGs, then it is not determinizable. This allows us to conclude the decidability of determinization in \cref{thm: main}.

Conceptually, our approach can be viewed as a ``standard'' one when treating determinization of quantitative models, in the sense that considering gaps between runs generally characterizes when a deterministic equivalent exists~\cite{filiot2017delay,2f4c7b43623c4e64883e3032b3b05950}. The crux is showing that this condition is decidable. To this end, our work greatly differs from other works on weighted automata in that we establish the decidability of the condition. Technically, this involves careful analysis of the behaviors of runs under discounting.

\section{Finitely Many Recoverable Gaps Imply Determinizability}
\label{sec:fin gaps implies det}

The main result of this section is an adaptation of the determinization techniques in~\cite{boker_et_al:LIPIcs:2011:3224} from complete to general automata. While the construction itself is similar, the correctness proof requires finer analysis. We remark that in the case that $\cA$ is a complete NDA and all final weights are zero, the construction obtains a complete DDA with all final weights zero, thus generalizing the result in~\cite{boker_et_al:LIPIcs:2011:3224}.

\begin{lemma}
\label{lem: finite gaps -> det}
If an NDA $\cA$ has finitely many recoverable gaps, then it is determinizable.
\end{lemma}

\begin{proof}
Let $\cA=(\Sigma,Q,Q_0,\alpha,\delta,\val,\fval,\lambda)$ be an NDA with finitely many recoverable gaps. We construct a DDA $\cD=(\Sigma,Q_D,\{v_0\},\alpha_D,\delta_D,\val_D,\fval_D,\lambda)$ that is equivalent to $\cA$. 

Since $\cA$ has finitely many recoverable gaps, there exists a bound $B\in \bbN$ on the size of those gaps.
The states of $\cD$ are then $Q_D=\{0,\ldots,B,\infty\}^Q$. 
Intuitively, a run of $\cD$ tracks, for each $q\in Q$, the gap between the minimal run of $\cA$ on $w$ ending in $q$ and the minimal run on $w$ overall. When this gap becomes too large to be recoverable, the states corresponding to the higher run are assigned $\infty$. For $v\in Q_D$ and $q\in Q$, we denote by $(v_q)$ the entry in $v$ corresponding to $q$. The initial state is therefore $(v_0)_q=
\begin{cases}
    0 & q\in Q_0\\
    \infty & q\notin Q_0
\end{cases}$, assigning for each $q\in Q$ the weight of the minimal run of $\cA$ on the empty word ending in $q$.

We now turn to define  $\delta_D$. Intuitively, when taking a transition, $\cD$ First updates the vector entry of every state with the value of the minimal run on the new word ending in it, using the values specified in the last vector. Then, if the minimal entry is not $0$, the entries are shifted so that it becomes $0$, and the value subtracted from every entry is assigned to the transition weight. Finally, the entries are all multiplied by $\lambda$ to account for the word length. Thus, the actual value of the minimal run is exactly the value attained by $\cD$, and the vector entries correctly represent the normalized gaps. The construction is demonstrated in~\cref{fig: construction example}.
Formally:
\begin{itemize}
    \item For every $v\in Q_D$, and for every $\sigma\in \Sigma$ such that there exists $q\in Q$ with $v_q<\infty$ and $q$ has a $\sigma$-transition, define $u\in \{0, \ldots ,B,\infty\}^Q$ as follows.
    \begin{itemize}
        \item Define the intermediate vector $u'$: For every $q\in Q$,  $u'_q=\min_{q'\in Q}(v_{q'}+\val(q',\sigma,q))$, where $\val(q',\sigma,q)$ is regarded as $\infty$ if $(q',\sigma,q)\notin \delta$.
        \item Define $r=\min_{q\in Q}u'_q$, the offset of the vector from $0$. Note that $r$ is finite due to the requirement that there exists $q\in Q$ with $v_q<\infty$ and $q$ has a $\sigma$-transition.
        \item For every $q\in Q$, $u_q=\begin{cases}
            \lam(u'_q-r) & \lam(u'_q-r)\leq B\\
            \infty & \text{otherwise}
        \end{cases}$
    \end{itemize}
    
    Where $\infty$ is handled using the standard semantics. Note that $u\in \{0, \ldots ,B,\infty\}^Q$ as $\cA$ is integral. The manipulations done on the intermediate vector $u'_q$ when defining $u_q$ should be viewed as normalization -- first subtracting $r$ so that the gap represented by $u_q$ is with respect to the minimal run overall over $w$; then multiplying by $\lam$ to account for the length of $w$.  Note that the subtraction of $r$ also implies that $\min_{q\in Q}u_q=0$, as is expected since $\min_{q\in Q}\cA_{[Q_0\to q]}(w)=\cA(w)$.
    \item We now introduce the transition $(v,\sigma,u)\in \delta_D$.
    \item We set $\val_D(v,\sigma,u)=r$. This can be viewed, together with the subtraction of $r$ from every entry of $u'$, as transferring the weight from each entry of $u'$ to the transition.
\end{itemize}

\begin{figure}
\centering 
\begin{subfigure}{0.45\textwidth}
\qquad
\begin{tikzpicture}[auto,node distance=2.8cm,scale=1]
        \tikzset{every state/.style={minimum size=20pt, inner sep=4}};
        \node (q0) [initial, accepting, state, initial text = {}] at (0,0) {$q_0$};
        \node (q1) [accepting, state] at (2,0) {$q_1$};
        \path [-stealth]
        (q0) edge node [above] {$a,0$} (q1)
        (q0) edge [loop below, distance=30,out=-45,in=225] node {$a,1$} (q0);
        \draw (0,-1.5) node{$b,0$};
    \end{tikzpicture}
    \caption{an NDA $\cA$.}
    \label{subfig:NDA_for_det_construction}
\end{subfigure}
    \hfill
    \centering
\begin{subfigure}{0.45\textwidth}
\qquad
\begin{tikzpicture}[auto,node distance=2.8cm,scale=1]
        \node (q0) [initial, accepting, state, initial text = {}] at (0,0) {$0,\infty$};
        \node (q1) [accepting, state] at (2,0) {$2,0$};
        \path [-stealth]
        (q0) edge [bend left,above] node {$a,0$} (q1)
        (q0) edge [loop below, distance=30,out=-45,in=225] node {$b,0$} (q0)
        (q1) edge [bend left,below] node {$b,2$} (q0)
        (q1) edge [loop below, distance=30,out=-45,in=225] node {$a,2$} (q1);
    \end{tikzpicture}
    \caption{An equivalent DDA $\cD$.}
    \label{subfig:DDA_for_det_construction}
\end{subfigure}
\caption{An example of an NDA $\cA$ (on the left) and the resulting DDA $\cD$ (on the right), with $\lambda=2$. The name of each state of $\cD$ corresponds to a vector whose first entry tracks $q_0$ and the second $q_1$. We demonstrate the construction using the $a$-transition from $(2,0)$ to itself. First we construct the intermediate vector $u'$: $(u')_{q_0}=\min(2+\val(q_0,a,q_0),0+\val(q_1,a,q_0))=\min(2+1,0+\infty)=3$ and $(u')_{q_1}=\min(2+\val(q_0,a,q_1),0+\val(q_1,a,q_1))=\min(2+0,0+\infty)=2$, and so $u'=(3,2)$. We then have $r=2$, which is assigned to the weight of the transition, and $u=2(3-2,2-2)=(2,0)$.}
\label{fig: construction example}
\end{figure}
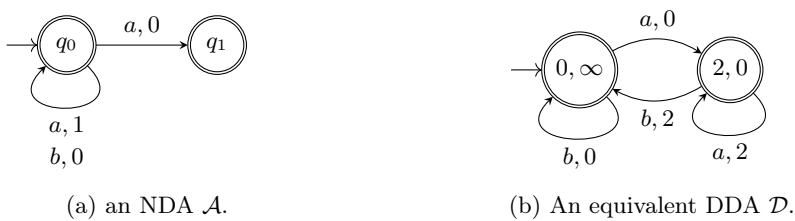

%Since $|\{0, \ldots ,B,\infty\}^Q|$ is finite, this process halts. 
We next define $\alpha_D$ and $\fval_D$. We set $\alpha_D$ to include every vector $v$ such that $v_q<\infty$ for some $q\in \alpha$. We note that the construction can be viewed as a generalization of the standard subset construction, where for a vector $v$, the states $q$ that satisfy $v_q<\infty$ represent the states that can be reached by $\cA$ when reading $w$, ignoring those states whose gap is unrecoverable. For $v\in \alpha_D$, we set $\fval_D(v)=\min_{q\in \alpha}(v_q+\fval(q))$. \cref{fig: construction example} depicts an example for an NDA and the DDA constructed from it (with no final weights). Note that we do not yet actually provide an algorithm for constructing $\cD$ from $\cA$, since that requires computing $B$.

The correctness of this construction is in the full proof (\cref{apx:finite gaps -> det}).
\hfill\qed\end{proof}

\section{Recoverable Gaps and $n$-separation}
\label{sec:reoverableGapsandSeparation}

\subsection{The relationship between TRGs and GRGs}
\label{sec: inf GRG to inf TRG}
We show that the characterization of $\cA$ as having finitely/infinitely many recoverable gaps can be reduced to TRGs. Since every TRG is a recoverable gap, the existence infinitely many TRGs trivially implies  the existence of infinitely many recoverable gaps. For the converse, it remains to show that the existence of infinitely many GRGs implies the existence of infinitely many TRGs.

% \subsection{A Large Gap is Equivalent to Infinitely Many Gaps}
% \label{sec:large gap to inf gaps}
% In this section we show that the existence of infinitely many TRGs is characterized by the existence of a (computable) large-enough TRGs. 

Consider a run $\rho=\rho_0 \ldots \rho_n$, and recall that $\val(\rho)$ is the weight of $\rho$ and that $\cM=2\frac{\lambda}{\lambda-1}m_{\cA}$, where $m_\cA$ is the maximal absolute value of a weight of a transition or a final weight in $\cA$. We denote by $\Gamma(\rho)=\undisc{n}\val(\rho)$ the normalized ``un-discounted'' value of $\rho$. 
For two runs $\rho^1,\rho^2$ on the same word $w$, we are interested in the value $\Gamma(\rho^1)-\Gamma(\rho^2)$, as it captures how far the runs are from each other, in the sense of how difficult it is to recover their gap. We claim that if two runs get too far from each other, the gap between them from that point on can only increase. Intuitively, this is because at each step the value is multiplied by $\lambda$, and so beyond a certain gap size, this multiplication separates the runs further even if their added values pull them closer before multiplying by $\lambda$. 

\begin{lemma}
\label{lem: a large gap has to increase}
Let $\rho^1=\rho^1_0, \ldots ,\rho^1_{n+1}$ and $\rho^2=\rho^2_0, \ldots ,\rho^2_{n+1}$ be two runs of $\cA$, such that $\Gamma(\rho^1_0, \ldots ,\rho^1_n)-\Gamma(\rho^2_0, \ldots ,\rho^2_n)>\cM$. Then $\Gamma(\rho^1_0, \ldots ,\rho^1_{n+1})-\Gamma(\rho^2_0, \ldots ,\rho^2_{n+1})>\Gamma(\rho^1_0, \ldots ,\rho^1_n)-\Gamma(\rho^2_0, \ldots ,\rho^2_n)$.
\end{lemma}
% \shtodo{Maybe appendix:}
% \begin{proof}
% We have $\Gamma(\rho_0^1, \ldots ,\rho_{n+1}^1) \geq \lam(\Gamma(\rho_0^1, \ldots ,\rho_n^1)-m_\cA)$, since in the worst case the last transition weighs $-m_{\cA}$. Similarly, $\Gamma(\rho_0^2, \ldots ,\rho_{n+1}^2) \leq \lam(\Gamma(\rho_0^2, \ldots ,\rho_n^2)+m_\cA)$. It follows that $\Gamma(\rho_0^1, \ldots ,\rho_{n+1}^1)-\Gamma(\rho_0^2, \ldots ,\rho_{n+1}^2)\geq \lam(\Gamma(\rho_0^1, \ldots ,\rho_n^1)-\Gamma(\rho_0^2, \ldots ,\rho_n^2)-2m_{\cA})$. Now, we have:

% \begin{align*}
% \Gamma(\rho_0^1, \ldots ,\rho_n^1)-\Gamma(\rho_0^2, \ldots ,\rho_n^2)> & 2\frac{\lambda}{\lambda-1}m_\cA\\
% (\lambda -1)(\Gamma(\rho_0^1, \ldots ,\rho_n^1)-\Gamma(\rho_0^2, \ldots ,\rho_n^2))> & 2\lambda m_\cA\\
% \lambda (\Gamma(\rho_0^1, \ldots ,\rho_n^1)-\Gamma(\rho_0^2, \ldots ,\rho_n^2))-2\lambda m_\cA> & \Gamma(\rho_0^1, \ldots ,\rho_n^1)-\Gamma(\rho_0^2, \ldots ,\rho_n^2)\\
% \lambda (\Gamma(\rho_0^1, \ldots ,\rho_n^1)-\Gamma(\rho_0^2, \ldots ,\rho_n^2)-2m_\cA)> & \Gamma(\rho_0^1, \ldots ,\rho_n^1)-\Gamma(\rho_0^2, \ldots ,\rho_n^2) \quad\qed\\
% \end{align*}
% \end{proof}

In particular, once the gap between $\rho^1,\rho^2$ is larger than $\cM$, concatenating any runs to $\rho^1,\rho^2$ can only increase the gap and therefore cannot result in $\rho^1$ overtaking $\rho^2$.

\begin{corollary}
\label{cor: Udi's constant unrecoverablility}
Let $\rho^1=\rho_0^1, \ldots ,\rho_n^1$ and $\rho^2=\rho_0^2, \ldots ,\rho_n^2$ be two runs such that $\val(\rho^1)\leq \val(\rho^2)$. Then for every $0\leq i\leq n$, it holds that $\Gamma(\rho^1_0, \ldots ,\rho^1_i)-\Gamma(\rho^2_0, \ldots ,\rho^2_i)\leq \cM$.
\end{corollary}

\begin{lemma}
\label{lem: inf GRG to inf TRG}
    If $\cA$ has infinitely many GRGs, it has infinitely many TRGs.
\end{lemma}
\begin{proof}
Intuitively, if we have a very large GRG on words $wz$, the upper run on $w$ eventually becomes minimal upon reading $z$. Since $\cA$ is trim, this run can be extended to an accepting state via a suffix $y$. Then, we can show that a minimal run on $wzy$ must form a large TRG with the original minimal run on $wz$ (albeit possible somewhat lower than the GRG).

Assume $\cA$ has infinitely many GRGs. We fix a number $G$ and show that $\cA$ has a TRG larger than $G$. We assume w.l.o.g. that $G\ge \cM$. There exists a GRG $(w,q_u,q_l)$ with respect to some $z$ such that $\gap{w}{q_u}{q_l}>2G$. Let $\rho=\rho_0, \ldots ,\rho_{|wz|}$ be a minimal run on $wz$ such that $q_u=\rho_{|w|}$, and let $q=\rho_{|wz|}$. Since $\cA$ is trim, there exists $y$ such that $\cA^*_{[q\to_f\alpha]}(y)<\infty$, and thus $\cA^*(wzy)<\infty$. Let $\rho'=\rho'_0, \ldots ,\rho'_{|wzy|}$ be a minimal accepting run on $wzy$, and let $q'_u=\rho'_{|w|}$. In particular, $(w,q'_u,q_l)$ is a TRG with respect to $zy$. It remains to show that $\gap{w}{q'_u}{q_l}>G$. Since $\rho$ is a minimal run on $wz$, in particular $\val(\rho) \le \val(\rho'_0,\ldots,\rho'_{|wz|})$, and by \cref{cor: Udi's constant unrecoverablility} we $\Gamma(\rho_0,\ldots,\rho_{|w|})-\Gamma(\rho'_0,\ldots,\rho'_{|w|}) \le \cM \le G$. Additionally, $\Gamma(\rho_0,\ldots,\rho_{|w|})=\undisc{|w|}\cA_{[Q_0\to q_u]}$. Hence,
\begin{align*}
    \gap{w}{q'_u}{q_l} & = \undisc{|w|}(\cA_{[Q_0\to q'_u]}(w)-\cA_{[Q_0\to q_l]}(w)) \\
    & \ge \Gamma(\rho'_0,\ldots,\rho'_{|w|})-\undisc{|w|}\cA_{Q_0\to q_l}(w) \\
    & = (\Gamma(\rho_0,\ldots,\rho_{|w|})-\undisc{|w|}\cA_{Q_0\to q_l}(w)) - (\Gamma(\rho_0,\ldots,\rho_{|w|})-\Gamma(\rho'_0,\ldots,\rho'_{|wz|})) \\
    & \ge \undisc{|w|}(\cA_{[Q_0\to q_u]}-\cA_{Q_0\to q_l}(w))-G \\
    & = \gap{w}{q_u}{q_l}-G \\
    & \ge G
\end{align*}
\hfill\qed\end{proof}

Since every recoverable gap is either a TRG or a GRG, we get the following:
\begin{corollary}
\label{cor: inf gaps iff inf TRG}
    $\cA$ has infinitely many recoverable gaps if and only if it has infinitely many TRGs.
\end{corollary}

\subsection{A Large Gap is Equivalent to Infinitely Many Gaps}
\label{sec:large gap to inf gaps}
In this section we show that the existence of infinitely many TRGs is characterized by the existence of a (computable) large-enough TRGs. 

In contrast to \cref{cor: Udi's constant unrecoverablility}, the gap between two runs cannot increase too much within a small number of steps. We capture the contra-positive of this, by showing that if two runs reach a large enough gap, then the runs have been far from each other for a long suffix.

\begin{lemma}
\label{lem: it takes time to reach a large gap}
Consider $n_\text{steps},n_\text{gap}\in \mathbb{N}$, there exists an effectively computable number $N$ such that for any two runs $\rho^1=\rho^1_0, \ldots ,\rho^1_n$ and $\rho^2=\rho^2_0, \ldots ,\rho^2_n$, if $\Gamma(\rho^1)-\Gamma(\rho^2)>N$ then $n>n_\text{steps}$ and $\Gamma(\rho^1_0, \ldots ,\rho^1_{n-n_\text{steps}})-\Gamma(\rho^2_0, \ldots ,\rho^2_{n-n_\text{steps}})>n_\text{gap}$.
\end{lemma}

We also need a version of \cref{cor: Udi's constant unrecoverablility} where the inequality between the weights of the runs includes final weights. We claim that concatenating any runs to runs that are far from each other cannot result in the lower run overtaking the upper run, including final weights:

\begin{lemma}
\label{lem: Udi's constant unrecoverablility - with final weights}
Let $\rho^u,\rho^l$ be two runs of $\cA$ on $w$, ending in states $q_u,q_l$ respectively, such that $\Gamma(\rho^u)-\Gamma(\rho^l)>\cM$. Let $\rho^{u_f},\rho^{l_f}$ be accepting runs on $z$ starting in $q_u,q_l$ respectively and ending in $q_{u_f},q_{u_l}$ respectively. Then $\val(\rho^u\rho^{u_f})+\disc{|wz|}\fval(q_{u_f})>\val(\rho^l\rho^{l_f})+\disc{|wz|}\fval(q_{l_f})$.
\end{lemma}
% \begin{proof}
% \begin{align*}
%     |\val(\rho^{u_f})+\disc{|z|}\fval(q_{u_f})| \leq & \Sigma_{i=0}^{|z|-1}\disc{i}|\val(\rho^{u_f}_{i},z_{i+1},\rho^{u_f}_{i+1})|+\disc{|z|}|\fval(q_{u_f})| \\
%     \leq & \Sigma_{i=0}^{|z|}\disc{i}m_\cA 
%     < \frac{\lam}{\lam-1}m_\cA
% \end{align*}
% And similarly, $|\val(\rho^{l_f})+\disc{|z|}\fval(q_{l_f})|<\frac{\lam}{\lam-1}m_\cA$. Therefore, the difference between the weights is lower than $2\frac{\lam}{\lam-1}m_\cA=\cM$ in absolute value. Now,
% \begin{align*}
%     & \val(\rho^u\rho^{u_f})+\disc{|wz|}\fval(q_{u_f})-(\val(\rho^l\rho^{l_f})+\disc{|wz|}\fval(q_{l_f})) \\
%     > & \val(\rho^u)-\val(\rho^l) - \disc{|w|}\cM > \disc{|w|}\cM - \disc{|w|}\cM = 0
% \end{align*}
% \hfill\qed\end{proof}

In particular, once a gap becomes too large, the only way to recover from it is if the lower run cannot continue at all.

\begin{lemma}
\label{lem: Udi's constant unrecoverablility - gaps}
Consider a TRG $(w,q_u,q_l)$ with respect to $z$ such that $\gap{w}{q_u}{q_l}>\cM$, then $\cA_{[q_u\to_f \alpha]}(z)<\infty$ and $\cA_{[q_l\to_f \alpha]}(z)=\infty$.
\end{lemma}
\begin{proof}
From the second condition in the definition of a TRG (\cref{def: recoverable gap}), we have $\cA_{[q_u\to_f \alpha]}(z)<\infty$. Let $\rho^u,\rho^l$ be minimal runs on $|w|$ ending in $q_u,q_l$ respectively. Assume by way of contradiction that $\cA_{[q_l\to_f \alpha]}(z)<\infty$, that is, there exists an accepting run $\rho^{l\prime}$ on $z$ starting in $q_l$. Let $\rho^{u\prime}$ be a minimal accepting run on $z$ starting in $q_u$. Since $\undisc{|w|}(\val(\rho^u)-\val(\rho^l))=\gap{w}{q_u}{q_l}>\cM$, \cref{lem: Udi's constant unrecoverablility - with final weights} contradicts the fact that $\rho^u\rho^{u\prime}$ is a minimal accepting run on $wz$ by the definition of a TRG.
\hfill\qed\end{proof}

We can now prove the main result of this section.
\begin{lemma}
\label{lem: inf gaps <-> large gap}
There exists an effectively computable number $N$ (depending on $\cA$) such that $\cA$ has infinitely many TRGs if and only if there exists a TRG $(w,q_u,q_l)$ such that $\gap{w}{q_u}{q_l}>N$.
\end{lemma}
\paragraph*{Proof Overview}
We start with an overview of the more complex direction -- the existence of a large TRG implies the existence of infinitely many TRGs. Assume that $(w,q_u,q_l)$ is a large TRG with respect to $z$. We consider two minimal runs $\rho^{q_u},\rho^{q_l}$ on $w$ ending in $q_u$ and $q_l$, respectively. These two runs end ``far'' from each other, so we can use \cref{lem: it takes time to reach a large gap} to claim that for a large enough $N$, they have already been far from each other for a while. Specifically, for the last $n_{\text{steps}}$ steps the gap between the runs was at least $n_{\text{gap}}$ for some large $n_{\text{steps}},n_{\text{gap}}$ that we choose to fit our needs. 

We now look for two indices $i<j$ among the last $n_\text{steps}$ indices of $w$ such that pumping the infix of $w$ between $i$ and $j$ generates words that induce unboundedly large TRGs. To do so, we choose $n_{\text{steps}}$ such that $Q$ can be partitioned into two sets of states -- an upper set $U$ and a lower set $L$, that are far from each other and "separate" the runs $\rho^{q_u},\rho^{q_l}$ at step $i$. In particular, pumping the infix does not interleave the runs, and maintains the growing gap. The above is depicted in \cref{fig: gap pumping}. We require the following properties:
\begin{enumerate}
    \item Every two runs on the prefix $w_1 \cdots w_i$ of $w$ ending in $U$ and in $L$, respectively, that are minimal runs ending in their respective states, are far enough from each other to satisfy the condition of \cref{lem: a large gap has to increase};
    \item Every run on $w$ that is minimal among the runs ending in $q_u$ has to visit $U$ at the $i$'th step;
    \item Every run on $w$ that is minimal among the runs ending in $q_l$ has to visit $L$ at the $i$'th step;
\end{enumerate}
As we show, finding such a partition is possible by choosing $n_\text{gap}=(|Q|-1)\cM$.

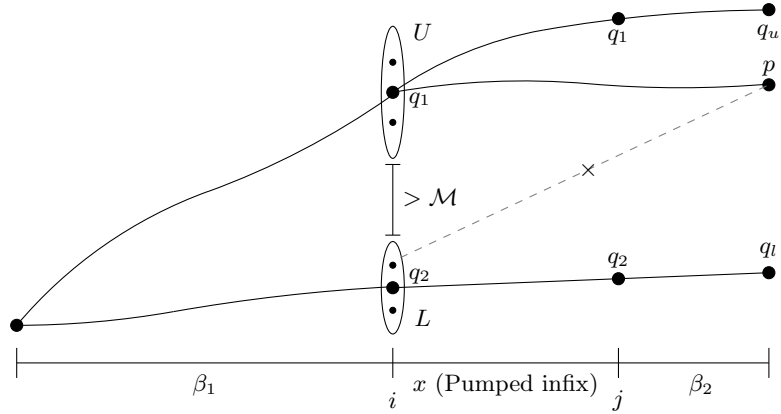
\begin{figure}[ht]
\begin{center}

\begin{tikzpicture}[scale=0.2]

\draw [black] (0,0) -- (50,0);
\draw [black] (0,-1) -- (0,1);
\draw [black] (25,-1) -- (25,1);
\draw [black] (40,-1) -- (40,1);
\draw [black] (50,-1) -- (50,1);
\draw (25,-2.5) node {$i$};
\draw (40,-2.5) node {$j$};
\draw (12.5,-1.5) node {$\beta_1$};
\draw (32.5,-1.5) node {$x$ (Pumped infix)};
\draw (45.5,-1.5) node {$\beta_2$};
\draw [black,fill=black] (0,2.5) circle (0.4);

\draw [black,fill=black] (25,5) circle (0.4);
\draw [black] (0,2.5) arc (270:280:60);
\draw [black] (10.42,3.41) arc (100:93:120);
\draw [black,fill=black] (25,18) circle (0.4);
\draw [black] (0,2.5) arc (140:110:30);
\draw [black] (12.72,11.4) arc (290:305:52);

\node[ellipse, draw, minimum width = 2, minimum height = 35] (e) at (25,5) {};
\draw [black,fill=black] (25,3.5) circle (0.2);
\draw [black,fill=black] (25,6.5) circle (0.2);
\draw (27,3) node {$L$};
\draw (26.8,5.8) node {$q_2$};

\node[ellipse, draw, minimum width = 2, minimum height = 50] (e2) at (25,18) {};
\draw [black,fill=black] (25,20) circle (0.2);
\draw [black,fill=black] (25,16) circle (0.2);
\draw (27,22) node {$U$};
\draw (26.8,17.5) node {$q_1$};

\draw [black] (25,8.5) -- (25,13.2);
\draw [black] (24.5,8.5) -- (25.5,8.5);
\draw [black] (24.5,13.2) -- (25.5,13.2);
\draw (27.5,11) node {$>\cM$};

\draw [black,fill=black] (50,6) circle (0.4);
\draw (50,7.5) node {$q_l$};
\draw [black] (25,5) -- (50,6);
\draw [black,fill=black] (40,5.6) circle (0.4);
\draw (40,6.8) node {$q_2$};

\draw [black,fill=black] (50,23.5) circle (0.4);
\draw (50,22) node {$q_u$};
\draw [black] (25,18) arc (125:100:25);
\draw [black] (35,22.14) arc (100:90:88);
\draw [black,fill=black] (40,22.9) circle (0.4);
\draw (40,21.7) node {$q_1$};

\draw [black] (25,18) arc (100:85:50);
\draw [black] (38.04,18.57) arc (265:275:70);
\draw [black,fill=black] (50,18.5) circle (0.4);
\draw (50,19.6) node {$p$};(0.4);
\draw [gray, dashed] (25.5,7) -- (50,18.5);
\draw (38,12.8) node {$\times$};

\end{tikzpicture}
\caption{At step $i$ of $\cA$'s run on $w$, the states are partitioned into an upper set $U$ and a lower set $L$ that are separated by a large gap. The runs $\rho^{q_u},\rho^{q_l}$ visit $U,L$ respectively, meaning the gap between them can only grow after step $i$. The indices $i,j$ are chosen such that both runs $\rho^{q_u},\rho^{q_l}$ repeat states and the sets of ancestors Anc$_q(i)$,Anc$_q(j)$ are identical for each $q\in Q$. The state $p$, which is visited after reading a pumped word $w^{(*)}$ by a minimal run on $w^{(*)}z$, is not reachable from $L$ on any of the pumped suffixes.}
\label{fig: gap pumping}
\end{center}
\end{figure}

Next, we show that in fact, $U$ and $L$ induce a certain separation of the run \emph{trees} emanating from them on the pumped words. Specifically, we show that:
\begin{enumerate}[label={(\roman*)}]
    \item There exist runs of $\cA$ on the pumped words (denoted $w^{(*)}$) ending in $q_u,q_l$.
    \item Every run on $w^{(*)}$ (ending in any state) that is a prefix of a minimal run on $w^{(*)}z$ has to visit $U$ at the $i$'th step. That is, a variant of Condition (2), where instead of $q_u$ we consider any state $p$ reached after reading $w^{(*)}$ along a minimal run on $w^{(*)}z$. 
    \item Condition (3) above holds not only for $w$ but for the pumped words $w^{(*)}$ as well.
\end{enumerate}
Note that (ii) and (iii) imply that runs on $w^{(*)}$ also induce a TRG.

From this, it follows from Condition 1 and~\cref{lem: a large gap has to increase} that the pumped words induce unboundedly large TRGs.

In order to ensure (i), we require $n_\text{steps}\ge |Q|^2$ (which is the length of the ``large gaps'' suffix) such that both runs $\rho^{q_u},\rho^{q_l}$ must repeat their pair of respective states at some indices $i,j$. Consequently, the runs $\rho^{q_u},\rho^{q_l}$ can be pumped to achieve the desired runs.

To ensure (ii), it follows from \cref{cor: Udi's constant unrecoverablility} and the fact that $\rho^{q_u}$ is a prefix of a minimal run on $wz$ that any state $p$ reached along a  run on $w^{(*)}z$ after reading $w^{(*)}$ is not reachable from $L$ when reading $w_{i+1}\cdots w_{|w|}$, and we want to ensure that $p$ is not reachable from $L$ when reading the pumped suffix as well. For that, for each state $q$ and for each index of $w$ we consider the set of states $\text{Anc}_q(i)$ from which $q$ is reachable when reading the respective suffix (from index $i$), called the \emph{ancestors} of $q$ at index $i$, and it is enough to require that for each state $q$ this set is identical for indices $i$ and $j$. This, in turn, requires to increase $n_{\text{steps}}$ by a factor of $2^{|Q|^2}$. Combined with the previous requirement on $i,j$, we choose $n_\text{steps}=|Q|^22^{|Q|^2}$.

Finally, for condition (3) in (iii), we use the fact that there exists a run ending in $q_l$ that visits $L$ at the $i$'th step (namely the pumped run) and apply \cref{cor: Udi's constant unrecoverablility}. Indeed, any run that does not visit $L$ at the $i$'th step must visit $U$ instead, and by \cref{cor: Udi's constant unrecoverablility} and the gap between $U$ and $L$, it must be larger than the run we have that visits $L$ and therefore not minimal.
\hfill \qed
\begin{proof}[of~\cref{lem: inf gaps <-> large gap}]
Consider runs $\rho_0^1,\ldots,\rho_n^1$ and $\rho_0^2,\ldots,\rho_n^2$. From \cref{lem: it takes time to reach a large gap}, we can effectively compute $N$ such that if $\Gamma(\rho_0^1,\ldots ,\rho_n^1)-\Gamma(\rho_0^2,\ldots,\rho_n^2)>N$, then $n>|Q|^{2}2^{|Q|^2}$ and $\Gamma(\rho_0^1, \ldots ,\rho_{n-|Q|^22^{|Q|^2}}^1)-\Gamma(\rho_0^2, \ldots ,\rho_{n-|Q|^22^{|Q|^2}}^2)>(|Q|-1)\cM$. 

Assume that $(w,q_u,q_l)$ is a TRG with respect to $z$ and $\gap{w}{q_u}{q_l}>N$. Let $\rho^{q_{u}}=\rho_{0}^{q_{u}} \ldots \rho_{|w|}^{q_{u}}$
be a run on $w$ ending in $q_{u}$ that is minimal among the runs on $w$ ending in $q_u$, and similarly $\rho^{q_{l}}=\rho_{0}^{q_{l}} \ldots \rho_{|w|}^{q_{l}}$ for $q_{l}$. Since these runs are minimal runs ending in their respective states, it holds that $\Gamma(\rho^{q_u})-\Gamma(\rho^{q_l})=\gap{w}{q_{u}}{q_{l}}>N$, and so we have $|w|>|Q|^{2}2^{|Q|^2}$ and $\Gamma(\rho_0^1, \ldots ,\rho_{|w|-|Q|^22^{|Q|^2}}^1)-\Gamma(\rho_0^2, \ldots ,\rho_{|w|-|Q|^22^{|Q|^2}}^2)>(|Q|-1)\cM$.

For every $q\in Q$ and $1\leq i\leq|w|$, let $\text{Anc}_q(i)=\{q'\in Q\mid\cA_{[q'\to q]}(w_{i+1} \cdots w_{|w|})<\infty\}$ be the set of \emph{ancestors} of $q$ at step $i$, i.e., states from which $q$ is reachable when reading the input $w_{i+1}, \cdots ,w_{|w|}$.
By the pigeonhole principle there exist $|w|-|Q|^{2}2^{|Q|^2}\leq i<j\leq|w|$
such that
\begin{itemize}
\item For every $q\in Q$, $\text{Anc}_q(i)=\text{Anc}_q(j)$, 
\item $\rho_{i}^{q_{u}}=\rho_{j}^{q_{u}}$ and $\rho_{i}^{q_{l}}=\rho_{j}^{q_{l}}$.
\end{itemize}

Write $w=\beta_1x\beta_2$ where $\beta_1=w_{1} \cdots w_{i}$, $x=w_{i+1} \cdots w_{j}$ and $\beta_2=w_{j+1} \cdots w_{|w|}$. We now turn to show that by pumping $x$, we can obtain unboundedly large TRGs.

Let $k\in\mathbb{N}$.  We can easily show that $\beta_1x^k\beta_2$ induces unboundedly large gaps between $q_u$ and $q_l$, but that would not be sufficient: We also need those gaps to be recoverable with respect to $z$, that is, the minimal run on $\beta_1x^k\beta_2z$ has to visit $q_u$ after reading $\beta_1x^k\beta_2$. However, this is not necessarily true: It can visit a different state, and we need to show that that state is also far enough from $q_l$. 
The runs $\rho_{0}^{q_{u}} \ldots \rho_{i}^{q_{u}}(\rho_{i+1}^{q_{u}} \ldots \rho_{j}^{q_{u}})^{k}\rho_{j+1}^{q_{u}} \ldots \rho_{|w|}^{q_{u}}$ and $\rho_{0}^{q_{l}} \ldots \rho_{i}^{q_{l}}(\rho_{i+1}^{q_{l}} \ldots \rho_{j}^{q_{l}})^{k}\rho_{j+1}^{q_{l}} \ldots \rho_{|w|}^{q_{l}}$ are runs on $\beta_1x^k\beta_2$ ending in $q_{u},q_{l}$ respectively. 
In particular, since there exists a run on $z$ starting in $q_u$, we have that $\cA$ has a run on $\beta_1x^k\beta_2z$. Let $\rho$ be minimal among those runs, and let $q_{\min_k}$ be the state $\rho$ visits after reading $\beta_1x^k\beta_2$.  Let $\rho^{q_{\min_k},k},\rho^{q_{l},k}$ be runs on $\beta_1 x^k \beta_2$ that are minimal among the runs ending in $q_{\min_k},q_l$ respectively.
Note that $\rho^{q_{\min_k},k}$ can be obtained as a prefix of $\rho$, since $\rho$ is minimal.
Then we have $\gap{\beta_1 x^{k}\beta_2}{q_{\min_k}}{q_{l}}=\Gamma(\rho^{q_{\min_k},k})-\Gamma(\rho^{q_{l},k})$, and it remains to show that $\Gamma(\rho^{q_{\min_k},k})-\Gamma(\rho^{q_{l},k})$ can get unboundedly large for a large enough $k$.

We already know that the runs $\rho^{q_u},\rho^{q_l}$ are far enough from each other at their $i$'th step to satisfy the condition of \cref{lem: a large gap has to increase}, and we want to show that the same is true for $\rho^{q_{\min_k},k},\rho^{q_l,k}$. 

To do so, we intuitively show that after reading $\beta_1$, the runs $\rho^{q_u},\rho^{q_l}$ have become so far apart that they now stem from disjoint sets of states with a large gap between them. Formally, consider the sets
\begin{align*}
U'&=\{ q\in Q\mid\text{there exists a run \ensuremath{\rho} on \ensuremath{w} with \ensuremath{\val(\rho)=\cA_{[Q_0\to q_u]}(w)} and \ensuremath{\rho_{i}=q}}\}\\
L'&=\{ q\in Q\mid\text{there exists a run \ensuremath{\rho} on \ensuremath{w} with \ensuremath{\val(\rho)=\cA_{[Q_0\to q_l]}(w)} and \ensuremath{\rho_{i}=q}}\}
\end{align*}
That is, $U'$ (resp. $L'$) is the set of states that appear at step $i$ in a minimal run to $q_u$ (resp. $q_l$).
For every $q\in Q$, let $v(q)=\undisc i \cA_{[Q_0\to q]}(\beta_1)$ be the ``undiscounted'' value of a minimal run of $\cA$ on $\beta_1$ ending in $q$. Then, from \cref{lem: it takes time to reach a large gap} and the constants we chose, for every $q_{u}'\in U',q_{l}'\in L'$ we have $v(q_{u}')-v(q_{l}')>(|Q|-1)\cM$.

In particular, there is a partition of $Q$ into two sets $U,L$
such that $U'\subseteq U,L'\subseteq L$ and $v(p)-v(q)>\cM$ for every $p\in U$ and $q\in L$. Indeed, otherwise the maximal gap between two states is less than $(|Q|-1)\cM$.
We next show that (i) $\rho_i^{q_{\min_k},k}\in U$, and (ii) $\rho_i^{q_l,k}\in L$.

For (i), we note that $\cA_{[L\to q_{\min_k}]}(x\beta_2)=\infty$: Otherwise, since $\cA_{[q_{\min_k}\to_f \alpha]}(z)<\infty$, there exists an accepting run on $wz$ that visits $L$ after reading $\beta_1$ and $q_{\min_k}$ after reading $w$. By \cref{lem: Udi's constant unrecoverablility - with final weights}, such a run must be of lower weight than any run that visits $U$ after reading $\beta_1$, in contradiction to the fact that $\rho^{q_u}$ is a prefix of a minimal accepting run on $wz$ by the definition of a TRG. Since $\text{Anc}_{q_{\min_k}}(i)=\text{Anc}_{q_{\min_k}}(j)$,
we also have that $\cA_{[L\to q_{\min_k}]}(x^{k}\beta_2)=\infty$. In particular, $\rho_{i}^{q_{\min_k},k}\in U$.

For (ii), the run $\rho_{0}^{q_{l}} \ldots \rho_{i}^{q_{l}}(\rho_{i+1}^{q_{l}} \ldots \rho_{j}^{q_{l}})^{k}\rho_{j+1}^{q_{l}} \ldots \rho_{|w|}^{q_{l}}$ satisfies $\rho_{i}^{q_{l}}\in L$ (since it is in $L'$), and in particular $\cA_{[L\to q_l]}(x^k\beta_2)<\infty$. By \cref{cor: Udi's constant unrecoverablility}, any run whose $i$'th state is in $U$ results in a higher weight than any run whose $i$'th state is in $L$, and so since  $\rho^{q_{l},k}$ is minimal we have $\rho_{i}^{q_{l},k}\in L$.

It remains to show that the runs $\rho^{q_{\min_k},k},\rho^{q_l,k}$, being far from each other at the $i$'th step, get unboundedly far from each other as $k$ increases. Let $f_{u},f_{l}:\{ i,i+1, \ldots \} \to\mathbb{N}$ be defined
as follows:
\begin{itemize}
\item $f_{u}(i)=\min_{q\in U}v(q)$
\item For $m\geq i$, $f_{u}(m+1)=\lam(f_{u}(m)-m_{\cA})$
\item $f_l(i)=\max_{q\in L}v(q)$
\item For $m\geq i$, $f_l(m+1)=\lam(f_l(m)+m_{\cA})$
\end{itemize}
Intuitively, $f_u$ (resp. $f_l$) represents a lower (resp. upper) bound on the "undiscounted" weight of runs visiting $U$ (resp. $L$) in their $i$'th step.
That is, for every $m\geq i$ we have $\Gamma(\rho_{0}^{q_{\min_k},k} \ldots \rho_{m}^{q_{\min_k},k})\geq f_u(m)$, and $\Gamma(\rho_{0}^{q_{l},k} \ldots \rho_{m}^{q_{l},k})\leq f_l(m)$.
Additionally, $f_u(i)-f_l(i)>\cM$ and so the function $f_u(m)-f_l(m)$ increases
with $m$. Thus, for every $M\in\mathbb{N}$, taking a large
enough $k$, we can obtain $\Gamma(\rho^{q_{\min_k},k})-\Gamma(\rho^{q_{l},k})\geq f_u(|\beta_1 x^{k}\beta_2|)-f_l(|\beta_1 x^{k}\beta_2|)>M$. 
This concludes the proof that if $\cA$ has a large TRG, then it has infinitely many TRGs.

For the converse direction, assume $\cA$ has infinitely many TRGs. Since $\cA$ is integral, the term $\undisc{|w|}(\cA_{[Q_0\to p]}(w)-\cA_{[Q_0\to q]}(w))$
is always an integer, therefore infinitely many TRGs imply the existence of unboundedly large TRGs, and in particular one larger than $N$.
\hfill\qed
\end{proof}

\begin{remark}
\label{rem:bound N}
Following the arguments in the proofs of \cref{lem: it takes time to reach a large gap,lem: inf gaps <-> large gap}, the number $N$ provided by \cref{lem: inf gaps <-> large gap} equals $\undisc{|Q|^22^{|Q|^2}}((|Q|-1)\cM-\cM)+\cM$. We denote this value by $\cN$.
\end{remark}

\subsection{A Large Gap is Equivalent to Separation}
\label{sec:gap to separation}
Recall that a gap refers to minimal runs that end in two specific states, but ignores the remaining states (to an extent). A more ``holistic'' view of gaps is via separations (\cref{def: separation property}). In this section we show that the two views are equivalent, and that both characterize when $\cA$ has infinitely many gaps.

\begin{lemma}
\label{lem: gap <-> separation}
$\cA$ has a TRG larger than $\cN$ if and only if $\cA$ has the $\cN$-separation property.
\end{lemma}

\begin{proof}
Assume that $\cA$ has a TRG larger than $\cN$. By \cref{lem: inf gaps <-> large gap}, there exist  unboundedly large TRGs, and in particular there exists a TRG $(w,q_u,q_l)$ with $\gap{w}{q_u}{q_l}>(|Q|-1)\cN$.

Intuitively, when ordering the states by the weight of the minimal run that reaches each state, such a gap implies a gap of at least $\cN$ between two successive states, leading to the desired partition. We then claim that the sets are separated by the same suffix $z$ that separates the states from the original gap.

Write $Q=\{q_1, \ldots ,q_{|Q|}\}$ such that $\cA_{[Q_0\to q_1]}(w)\leq \ldots \leq\cA_{[Q_0\to q_{|Q|}]}(w)$ (recall that if there is no run on $w$ ending in $q$, then $\cA_{[Q_0\to q]}(w)=\infty$), and let $i_l<i_u$ be indices such that $q_l=q_{i_l},q_u=q_{i_u}$. Then there exists $j\in \{i_l, \ldots ,i_{u-1}\}$ such that $\cA_{[Q_0\to q_{j+1}]}(w)-\cA_{[Q_0\to q_j]}(w)>\cN$. 
Let $U=\{q_{j+1}, \ldots ,q_{|Q|}\}$ and $L=\{q_1, \ldots ,q_j\}$. Then for every $q_u'\in U,q_l'\in L$ we have $\undisc{|w|}(\cA_{[Q_0\to q_u']}(w)-\cA_{[Q_0\to q_l']}(w))>\cN$.

Consider $z\in \Sigma^*$ such that $(w,q_u,q_l)$ is a TRG with respect to $z$. Note that $\cN>\cM$, and so it follows from \cref{lem: Udi's constant unrecoverablility - with final weights} that $\cA_{[q_l'\to_f \alpha]}(z)=\infty$ for every $q_l'\in L$. Indeed, if $\cA$ had an accepting run on $z$ starting in $q_l'$, concatenating it to a minimal run on $w$ ending in $q_l'$ would result in an accepting run of lower weight than any run on $wz$ that visits $q_u$ after reading $w$, contradicting the fact that $(w,q_u,q_l)$ is a TRG. Additionally, it follows from $(w,q_u,q_l)$ being a TRG that there exists a minimal accepting run on $wz$ that visits $q_u$ after reading $w$. Then $(w,q_u,q_l')$ is a TRG with respect to $z$, and so $w$ has the $\cN$-separation property with respect to $(U,L,q_u,z)$.

For the converse direction, assume that $w$ has the $\cN$-separation property with respect to some $(U,L,q_u,z)$. In particular, $\cA_{[Q_0\to q_u]}(w)+\disc{|w|}\cA_{[q_u\to_f \alpha]}(z)<\infty$. Let $q_u'\in U$ be such that $\cA_{[Q_0\to q_u']}(w)+\disc{|w|}\cA_{[q_u'\to_f \alpha]}(z)$ is minimal. Let some $q_l'\in L$. Then $(w,q_u',q_l')$ is a TRG with respect to $z$, and it is larger than $\cN$, as needed.
\hfill\qed\end{proof}

\section{Bounding the Witnesses for Separation}
\label{sec:small model}
In \cref{sec:reoverableGapsandSeparation} we show that $\cA$ has infinitely many TRGs if and only if there exists a word $w$ with the $\cN$-separation property. Expanding \cref{def: separation property}, this happens if and only if there exist a partition of $Q$ into two sets $U,L$ and there exist words $w,z$ that ``separate'' $U$ from $L$. In this section we can bound the length of such minimal $w,z$. We start with $w$ (see~\cref{apx: short w} for the full proof).

\begin{lemma} 
\label{lem: short w}
Let $C=\frac{\lambda}{\lambda-1}(\cN|Q|+2m_{\cA})$. Assume that $w$ has the $\cN$-separation property for some $w\in \Sigma^*$. Then there exists $w'$ such that $w'$ has the $\cN$-separation property and $|w'|\leq(C+2)^{|Q|}$.
\end{lemma}
\begin{proof}[Sketch]
Assume that $w$ has the $\cN$-separation property
with respect to\emph{ $(U,L,q_u,z)$}.

We start by using an identical construction to that of \cref{lem: finite gaps -> det}, with bound $C$, in order to define a sequence of vectors $v_0,\ldots,v_{|w|}$ with $v_i\in \{0,\ldots,C,\infty\}^Q$ for every $0\le i\le |w|$ that, intuitively, keep track of the runs of $\cA$ on $w$, as follows. 
\begin{itemize}
\item For every $q\in Q$ set $(v_{0})_q=\begin{cases}
    0 & q\in Q_0\\
    \infty & \text{otherwise}
\end{cases}$
\item For every $i>0,q\in Q$ let $v_{i,q}'=\min_{q'\in Q}((v_{i-1})_{q'}+\val(q',w_i,q))$, where $\val(q',\sigma,q)$ is regarded as $\infty$ if $(q',\sigma,q)\notin \delta$ (the $v'_{i,q}$ are ``intermediate'' values).
\item For every $i>0$ let $r_i=\min_{q\in Q}v_{i,q}'$ (the $r_i$ are the offset of the vector from $0$). 
\item For every $i>0,q\in Q$ set $(v_{i})_q=\begin{cases}
    \lambda(v_{i,q}'-r_i) & \lambda(v_{i,q}'-r_i)\leq C\\
    \infty & \text{otherwise}
\end{cases}$
\end{itemize}

Recall that intuitively, $(v_i)$ tracks, for each $q\in Q$, the gap between the minimal run on $w_1 \cdots w_i$ ending in $q$ and the minimal run on this prefix overall. When this gap becomes large enough that recovering from it implies the existence of $\cN$-separation, it is denoted $\infty$.

Denote the normalized difference $\undisc{i} (\cA_{[Q_0\to q]}(w_{1} \cdots w_{i})-\cA(w_{1} \cdots w_{i}))$ by $\Delta_{q,i}(w)$. It is easy to show that $v_i$ keeps the correct weight of runs whose gap from the minimal one remains always under $C$. However, if a gap of a run goes over $C$ but then comes back down, then $v_i$ no longer tracks it correctly. To account for this, we claim that since $w$ has the $\cN$-separation property, for every $q,i$ at least one of the following must
hold:
% \begin{itemize}
% \item $(v_{i})_q=\begin{cases}
% \undisc{i} (\cA_{[Q_0\to q]}(w_{1} \cdots w_{i})-\cA(w_{1} \cdots w_{i})) & \undisc i (\cA_{[Q_0\to q]}(w_{1} \cdots w_{i})-\cA(w_{1} \cdots w_{i}))\leq C\\
% \infty & \text{otherwise}
% \end{cases}$.
% \ndtodo{I think the only way to fix this width (which repeats later) is to give some short notation to that expression, does that make sense?}
% \item There exists $i'<i$ such that $w_{1} \cdots w_{i'}$ has the $\cN$-separation property.
% \end{itemize}
\begin{itemize}
\item $(v_{i})_q=\begin{cases}
\Delta_{q,i}(w) & \Delta_{q,i}(w)\leq C\\
\infty & \text{otherwise}
\end{cases}$.
\item There exists $i'<i$ such that $w_{1} \cdots w_{i'}$ has the $\cN$-separation property.
\end{itemize}
That is, either $v_i$ tracks the runs correctly, or there is some shorter prefix that already has the $\cN$-separation property.

The proof is by induction on $i$, with the only problematic case arising when $(v_{i-1})_{q'}=\infty$, and so the information about the exact value of the gap represented by $(v_{i-1})_{q'}$ is gone. We consider the normalization value $r_i$ (i.e., the offset of the minimal run from $0$): if $r_i$ is small, then the gap represented by $(v_i)_q$ is still very large, and we show that marking it as $\infty$ is sound. Otherwise, if $r_i$ is large, then the above gap might indeed be wrongly marked as $\infty$. However, we show that in this case, $r_i$ is so large that we can actually obtain an $\cN$-separation property ``below'' $r_i$, using a shorter witness. More precisely:
\begin{itemize}
\item If $(v_{i-1})_{q'}=\infty$ and $r_i\leq C\frac{\lambda-1}{\lambda}-m_{\cA}$, then since $(v_{i-1})_{q'}=\infty$, we have $(v_i)_q=\infty$.
It remains to show that $\undisc i(\cA_{[Q_0\to q]}(w_{1} \cdots w_{i})-\cA(w_{1} \cdots w_{i}))>C$.
Indeed, 
\begin{align*}
& \undisc i(\cA_{[Q_0\to q]}(w_{1} \cdots w_{i})-\cA(w_{1} \cdots w_{i})) \\
\geq & \undisc{i}(\cA_{[Q_0\to q']}(w_{1} \cdots w_{i-1})-\cA(w_{1} \cdots w_{i-1})-(m_{\cA}+r_i)\cdot\disc{(i-1)})\\
= & \lam(\undisc{i-1}(\cA_{[Q_0\to q']}(w_1 \cdots w_{i-1})-\cA(w_1 \cdots w_{i-1}))-r_i-m_{\cA})\\
> & \lam(C-(C\frac{\lam-1}{\lam}-m_{\cA})-m_{\cA}) = \lam(\lami C+m_{\cA}-m_{\cA}) > C
\end{align*}
where the first transition follows by observing that when reading $w_i$, in the worst case, the weight of a specific run can decrease by $\disc{(i-1)}m_{\cA}$, and the overall weight of the word can increase by $\disc{(i-1)}r_i$.
%; indeed, let $q_m$ be a state such that $v_{i,q_m}'$ is minimal (and therefore equals $r_i$), and let $q_m'$ be a state such that $v_{i,q_m}'=((v_{i-1})_{q_m'}+\val(q_m',w_i,q_m))$. Then
%\begin{align*}
%&\cA(w_1 \cdots w_i) \leq  \cA_{[Q_0\to q_m]}(w_1 \cdots w_i) \\
%\leq & \cA_{[Q_0\to q'_m]}(w_1 \cdots w_{i-1}) + \disc{i}\val(q_m',w_i,q_m) \\
%\overset{(*)}{=} & \cA(w_1 \cdots w_{i-1})+\disc{(i-1)}(v_{i-1})_{q'_m}+\disc{i}\val(q_m',w_i,q_m) \\
%5< & \cA(w_1 \cdots w_{i-1})+\disc{(i-1)}((v_{i-1})_{q'_m}+\val(q_m',w_i,q_m)) \\
%= & \cA(w_1 \cdots w_{i-1})+\disc{(i-1)}r_1
%\end{align*}
%Where $(*)$ is due to the induction hypothesis.

\item $r_i>C\frac{\lambda-1}{\lambda}-m_{\cA}$. This is only possible
if for every $q_l$ such that $(v_{i-1})_{q_l}<C\frac{\lambda-1}{\lambda}-2m_{\cA}=\cN|Q|$,
$q_{l}$ has no $w_{i}$-transition. Let $L''=\{ q_{l}\in Q\mid(v_{i-1})_{q_l}<\cN|Q|\} $.
Write $Q={q_1, \ldots ,q_{|Q|}}$
such that $(v_{i-1})_{q_1}\leq \ldots \leq(v_{i-1})_{q_{|Q|}}$, and so $L''=\{q_1, \ldots ,q_{|L''|}\}$. Since $w$ has the $\cN$-separation property, in particular $\cA$ has a run on $w$ and so $L''\subsetneq Q$. Then, there exists $1\leq r\leq |L''|$ such that $(v_{i-1})_{q_{r+1}}-(v_{i-1})_{q_r}>\cN$.
Let $U'=\{q_{r+1}, \ldots ,q_{|Q|}\},L'=\{q_1, \ldots ,q_r\}$,
and note that for every $q_l'\in L'$, $q_l'$
has no $w_{i}$-transition. For every $q_l'\in L',q_u'\in U'$, we have $\undisc{i-1}(\cA_{[Q_0\to q_u']}(w_1 \cdots w_{i-1})-\cA_{[Q_0\to q_l']}(w_1 \cdots w_{i-1})=(v_{i-1})_{q_u'}-(v_{i-1})_{q_l'}>\cN$.  Let $q_u'\in U'$ be such that $\cA_{[Q_0\to q_u']}(w_1 \cdots w_{i-1})+\disc{(i-1)}\cA_{[q_u'\to_f \alpha]}(w_i)$ is minimal. Then for every $q_l'\in L'$, $(w_1 \cdots w_{i-1},q_u',q_l')$ is a TRG with respect to $w_i$, and so $w_1 \cdots w_{i-1}$ has the $\cN$-separation property with respect to $(U',L',q_u',w_i)$, and we are done.
\end{itemize}

Now, it remains to show that if $|w|>(C+2)^{|Q|}$,
there exists $w'$ such that $|w'|<|w|$
and $w'$ has the $\cN$-separation property.

To this end, we use the induction hypothesis and the pigeonhole principle to remove an infix of $w$, and argue that the resulting word $w'$ also has the $\cN$-separation property with respect to some $(U',L',q_u')$: Either all of the minimal runs ending in the states of $L$ have values far enough (below) of $C$, in which case $U',L'$ can be chosen to be $U,L$ respectively; or some state of $L$ attains a high value, in which case there must be a large gap between two consecutive states of $L$, and the resulting lower set can be chosen as $L'$. As for $q_u'$, it is simply enough to consider the state in $U'$ that the minimal run on $w'z$ visits after reading $w'$ (see~\cref{apx: short w} for the details).
\hfill\qed
\end{proof}

Next, we give a bound on the length of the minimal separating suffix $z$ from \cref{def: separation property}. Recall that by \cref{lem: Udi's constant unrecoverablility - gaps}, a large gap can only be a TRG if the smaller runs cannot continue at all. Following that, we can now limit the search to suffixes that separate runs in a Boolean sense (i.e., making one accept and another reject). This yields a bound from standard arguments about Boolean automata, as follows.
\begin{lemma}
\label{lem: short z}
Consider a word $w$ that has the $\cN$-separation property with respect to $(U,L,q_u,z)$. Then there exists $z'$ such that $w$ has the $\cN$-separation property with respect to $(U,L,q_u,z')$ and $|z'|\leq 2^{2|Q|}$.
\end{lemma}
% \shtodo{to appendix:}
% \begin{proof}
% By \cref{lem: Udi's constant unrecoverablility - gaps} and that fact that $\cN>\cM$, it is enough to show that if for two sets of states $U,L$ and a word $z$ it holds that $\cA_{[U\to_f \alpha]}(z)<\infty$ and $\cA_{[L\to_f \alpha]}(z)=\infty$, then the same is true for some $z'$ such that $|z'|\leq 2^{2|Q|}$. Consider the NFAs $\cA_u,\cA_l$ whose states and transitions are identical to those of $\cA$, with $U,L$ as sets of initial states, respectively. A word $z'$ satisfies $\cA_{[U\to_f \alpha]}(z')<\infty$
% and $\cA_{[L\to_f \alpha]}(z')=\infty$ if and only if $z'\in L(\cA_U)\backslash L(\cA_L)$. The claim is therefore true due to the fact that for two NFAs $\cA_1,\cA_2$ with sets of states $Q_1,Q_2$ and languages $\mathcal{L}_1,\mathcal{L}_2$, there exists $z'\in \mathcal{L}_1\backslash\mathcal{L}_2$ if and only if there exists such $z'$ of length at most $2^{|Q_1|+|Q_2|}$.
% \hfill\qed\end{proof}

\section{Determinizability of Integral NDAs is Decidable}
\label{sec:det is decidable}
%\subsection{Characterizing Determinizability by Means of Gaps}
In this section we establish the decidability of determinization. To this end, we start by completing the characterization of determinizable NDAs by means of gaps, and then use the results from previous sections to conclude the decidability of this characterization.

Recall that in \cref{lem: finite gaps -> det} we show that finitely many recoverable gaps imply determinizability. In this section we show the converse, thus completing the characterization of determinizable integral NDAs as exactly those that have finitely many TRGs (and thus, by \cref{cor: inf gaps iff inf TRG}, those that have finitely many recoverable gaps).

\begin{lemma}
\label{lem: det -> finite TRGs}
    If an NDA $\cA$ is determinizable, it has finitely many TRGs.
\end{lemma}
\begin{proof}
Let $\cD=(\Sigma,Q_D,\{v_0\},\alpha_D,\delta_D,\val_D,\fval_D,\lambda)$ be a DDA equivalent to $\cA$. Fix $q_u,q_l\in Q,z\in\Sigma^*$, and let
\[
    G_{q_u,q_l,z}=\{w \mid (w,q_u,q_l) \text{ is a TRG with respect to } z\}
\]
If $w\in G_{q_u,q_l,z}$ then in particular $\cD^*(wz)=\cA^*(wz)<\infty$ and $\cA_{[Q_0\to q_l]}(w)<\infty$. Since $\cA$ is trim, there exists $y\in\Sigma^*$ such that $\cA_{[q_l\to_f\alpha]}(y)<\infty$. Therefore, $\cD^*(wy)=\cA^*(wy)<\infty$. Denoting $q_D=\delta_D^*(v_0,w)$, we now have:
\begin{align*}
    \cD^*(wz) & = \cD_{[v_0\to q_\cD]}(w) + \disc{|w|}\cD_{[q_\cD\to_f \alpha]}(z) \\
    \cD^*(wy) & = \cD_{[v_0\to q_\cD]}(w) + \disc{|w|}\cD_{[q_\cD\to_f \alpha]}(y)
\end{align*}
and thus 
\[
    \cD^*(wz)-\cD^*(wy) \le \disc{|w|}\cM_\cD
\]
Where $\cM_\cD=2\frac{\lambda}{\lambda-1}m_{\cD}$ (recall that $m_{\cD}$ is the maximal absolute value of a weight of a transition or a final weight in $\cD$). Additionally, 
\begin{align*}
    \cA^*(wz) & = \cA_{[Q_0\to q_u]}(w) + \disc{|w|}\cA_{[q_u\to_f \alpha]}(z) \\
    \cA^*(wy) & \le \cA_{[Q_0\to q_l]}(w) + \disc{|w|}\cA_{[q_l\to_f \alpha]}(y)
\end{align*}
and thus
\begin{align*}
    \cA^*(wz)-\cA^*(wy) & \ge \cA_{[Q_0\to q_u]}(w)-\cA_{[Q_0\to q_l]}(w) - \disc{|w|}\cM \\
    & = \disc{|w|}(\gap{w}{q_u}{q_l}-\cM)
\end{align*}
Therefore, 
\begin{align*}
    \gap{w}{q_u}{q_l} & \le \undisc{|w|}(\cA^*(wz)-\cA^*(wy)) + \cM \\
    & = \undisc{|w|}(\cD^*(wz)-\cD^*(wy)) + \cM \\
    & \le \cM_\cD + \cM
\end{align*}
Hence, 
\[
    \{\gap{w}{q_u}{q_l}\mid (w,q_u,q_l) \text{ is a TRG}\} = 
    \cup_{q_u,q_l\in Q,z\in\Sigma^*}\{\gap{w}{q_u}{q_l}\mid w\in G_{q_u,q_l,z}\}
\]
is bounded by $\cM_\cD + \cM$, thus proving the claim.
\end{proof}

%\subsection{Determinizability of integral NDA is Decidable}
Consider an NDA $\cA$. By \cref{cor: inf gaps iff inf TRG,lem: finite gaps -> det,lem: det -> finite TRGs,lem: inf gaps <-> large gap,lem: gap <-> separation,lem: short w,lem: short z} we have that $\cA$ has an equivalent DDA if and only if for every  $w,z$ such that $|w|\leq(\frac{\lambda}{\lambda-1}(\cN|Q|+2m_{\cA})+2)^{|Q|}$
and $|z|\leq2^{2|Q|}$, it holds that $w$  does not have the $\cN$-separation property
with respect to $(U,L,q_u,z)$ for every $U,L,q_u$. Since the latter condition can be checked by traversing finitely many words and simulating the runs of $\cA$ on each of them, we can conclude our main result.

\begin{theorem}
\label{thm: main} The problem of whether an integral NDA has a deterministic equivalent is decidable.
\end{theorem}

\begin{remark}[Complexity of Determinization]
    \label{rem:complexity}
    Using the bounds on $w,z$, one can guess $w,z$ on-the-fly, while keeping track of the weights of minimal runs to all states, discarding those that go above $C$ as per \cref{lem: short w}, to check whether $\cA$ has the $\cN$-separation property. Since $\cN$ is double exponential in the size of $\cA$, this procedure can be done in $\mathsf{NEXPSPACE}=\mathsf{EXPSPACE}$. Thus, determinizability is in $\mathsf{EXPSPACE}$. For a lower bound, determinizability is also $\mathsf{PSPACE-hard}$ by a standard reduction from NFA universality. Tightening this gap is left open. Note that for lowering the upper bound, we would need a refined application of the pigeonhole principle in~\cref{lem: inf gaps <-> large gap}, which seems somewhat out of reach for the pumping argument. Conversely, for increasing the lower bound, we would need to show that using discounting we can somehow force a double-exponential blowup in determinization. While this might be within reach, no such examples are known for e.g., tropical weighted automata, suggesting that this may be very difficult.
\end{remark}

\bibliographystyle{splncs04}
\bibliography{bib}

\appendix

\section{Proofs}
\subsection{Proof of~\cref{lem: finite gaps -> det}}
\label{apx:finite gaps -> det}
%%%!!!!!!!!!!!!!!!%%%%%%%%%
%Note that this proof is without the figure. Merge carefully.
%%%!!!!!!!!!!!!!!!%%%%%%%%%
Let $\cA=(\Sigma,Q,Q_0,\alpha,\delta,\val,\fval,\lambda)$ be an NDA with finitely many recoverable gaps. We construct a DDA $\cD=(\Sigma,Q_D,\{v_0\},\alpha_D,\delta_D,\val_D,\fval_D,\lambda)$ that is equivalent to $\cA$. 

Since $\cA$ has finitely many recoverable gaps, there exists a bound $B\in \bbN$ on the size of those gaps.
The states of $\cD$ are then $Q_D=\{0,\ldots,B,\infty\}^Q$. 
Intuitively, a run of $\cD$ tracks, for each $q\in Q$, the gap between the minimal run of $\cA$ ending in $q$ and the minimal run overall. When this gap becomes too large to be recoverable, the states corresponding to the higher run are assigned $\infty$. The initial state is therefore $(v_0)_q=
\begin{cases}
    0 & q\in Q_0\\
    \infty & q\notin Q_0
\end{cases}$, assigning for each $q\in Q$ the weight of the minimal run of $\cA$ on the empty word ending in $q$.

We now turn to define  $\delta_D$. Intuitively, when taking a transition, $\cD$ updates the reachable states with their minimal-run values, normalized by $\lambda$. Then, if the minimal run overall is not of value $0$, then the entries are shifted so that it becomes 0, and the shift value is assigned to the transition. Thus, the actual value of the minimal run is exactly the value attained by $\cD$. The construction is demonstrated in~\cref{fig: construction example}.
Formally:
\begin{itemize}
    \item For every $v\in Q_D$, and for every $\sigma\in \Sigma$ such that there exists $q\in Q$ with $v_q<\infty$ and $q$ has a $\sigma$-transition, define $u\in \{0, \ldots ,B,\infty\}^Q$ as follows.
    \begin{itemize}
        \item Define the intermediate vector $u'$: For every $q\in Q$,  $u'_q=\min_{q'\in Q}(v_{q'}+\val(q',\sigma,q))$, where $\val(q',\sigma,q)$ is regarded as $\infty$ if $(q',\sigma,q)\notin \delta$.
        \item Define $r=\min_{q\in Q}u'_q$, the offset of the vector from $0$. Note that $r$ is finite due to the requirement that there exists $q\in Q$ with $v_q<\infty$ and $q$ has a $\sigma$-transition.
        \item For every $q\in Q$, $u_q=\begin{cases}
            \lam(u'_q-r) & \lam(u'_q-r)\leq B\\
            \infty & \text{otherwise}
        \end{cases}$
    \end{itemize}
    
    Where $\infty$ is handled using the standard semantics. Note that $u\in \{0, \ldots ,B,\infty\}^Q$ as $\cA$ is integral. The manipulations done on the intermediate vector $u'_q$ when defining $u_q$ should be viewed as normalization -- first subtracting $r$ so that the gap represented by $u_q$ is with respect to the minimal run overall; then multiplying by $\lam$ to account for the length of $w$.  Note that the subtraction of $r$ also implies that $\min_{q\in Q}u_q=0$, as is expected since $\min_{q\in Q}\cA_{[Q_0\to q]}(w)=\cA(w)$.
    \item We now introduce the transition $(v,\sigma,u)\in \delta_D$.
    \item We set $\val_D(v,\sigma,u)=r$. This can be viewed, together with the subtraction of $r$ from every entry of $u'$, as transferring the weight from each entry of $u'$ to the transition.
\end{itemize}

%Since $|\{0, \ldots ,B,\infty\}^Q|$ is finite, this process halts. 
We next define $\alpha_D$ and $\fval_D$. We set $\alpha_D$ to include every vector $v$ such that $v_q<\infty$ for some $q\in \alpha$. We note that the construction can be viewed as a generalization of the standard subset construction, where for a vector $v$, the states $q$ that satisfy $v_q<\infty$ represent the states that can be reached by $\cA$ when reading $w$, ignoring those states whose gap is unrecoverable. For $v\in \alpha_D$, we set $\fval_D(v)=\min_{q\in \alpha}(v_q+\fval(q))$. \cref{fig: construction example} depicts an example for an NDA and the DDA constructed from it (with no final weights). Note that we do not yet actually provide an algorithm for constructing $\cD$ from $\cA$, since that requires computing $B$.

It remains to show the correctness of the construction, that is, that $\cA^*(w)=\cD^*(w)$ for every $w\in \Sigma^*$. We first prove the following claim by induction on $|w|$.

For every $w\in \Sigma^*$:
\begin{enumerate}
    \item $\cA(w)=\cD(w)$.
    \item Assume $\cD(w)<\infty$. Let $u=\delta_D^*(v_0,w)$ (the state $\cD$ reaches after reading $w$). Then for every $q\in Q$, the following holds:
    \begin{itemize}
        \item If $\cA_{[Q_0\to q]}(w)=\infty$ then $u_q=\infty$.
        \item If $\cA_{[Q_0\to q]}(w)<\infty$, then the following conditions both hold:
        \begin{itemize}
            \item $u_q$ equals either $\undisc{|w|}(\cA_{[Q_0\to q]}(w)-\cA(w))$ or $\infty$.
            \item If $(w,q,p)$ is a recoverable gap, where $p$ is the last state of some minimal run of $\cA$ on $w$ (note that if there is more than one choice for such a $p$, that choice does not affect the size or recoverability of the gap), then $u_q=\undisc{|w|}(\cA_{[Q_0\to q]}(w)-\cA(w))$. That is, $\cD$ correctly tracks the gaps, except for maybe when they are unrecoverable and therefore are of no interest.
        \end{itemize}
    \end{itemize}
\end{enumerate}
Note that in the last case we do not require that $u_q=\undisc{|w|}(\cA_{[Q_0\to q]}(w)-\cA(w))$ whenever $\undisc{|w|}(\cA_{[Q_0\to q]}(w)-\cA(w))\leq B$, as we allow $u_q=\infty$ in that case, provided that the gap is unrecoverable. In the case that a gap is unrecoverable and is not tracked correctly, it is assigned $\infty$ and therefore does not interfere with future tracking of weights of minimal runs.

We continue with the inductive proof of the claim. For $|w|=0$, we have:
\begin{enumerate}
    \item $\cA(w)=\cD(w)=0$
    \item If $q\in Q_0$ then $\cA_{[Q_0\to q]}(w)=0$ and indeed $u_q=(v_0)_q=0$. If $q\notin Q_0$ then $\cA_{[Q_0\to q]}(w)=\infty$ and indeed $u_q=(v_0)_q=\infty$.
\end{enumerate}

Assume the hypothesis is true for $w$ and consider $w\sigma$ for some $\sigma\in\Sigma$. 
\begin{enumerate}
\item We first show that $\cA(w\sigma)<\infty\iff\cD(w\sigma)<\infty$. 

Assume that $\cD(w\sigma)<\infty$, that is, $\cD$ has a run on $w\sigma$. Let $v=\delta_D^*(v_0,w)$. Then $v$ has a $\sigma$-transition, meaning there exists $q\in Q$ such that $v_q<\infty$ and $q$ has a $\sigma$-transition. By the induction hypothesis, $v_q<\infty$ implies that $\cA_{[Q_0\to q]}(w)<\infty$, and since $q$ has a $\sigma$-transition we also have $\cA(w\sigma)<\infty$.

Now assume that $\cA(w\sigma)<\infty$. Then $\cA(w)<\infty$ and by the induction hypothesis $\cD(w)<\infty$, that is, $\cD$ has a run on $w$. Let $v=\delta_D^*(v_0,w)$. Let $\rho=\rho_1, \ldots ,\rho_{|w\sigma|}$ be a minimal run of $\cA$ on $w\sigma$, and denote $q_m=\rho_{|w\sigma|},q'_m=\rho_{|w|}$. Let $p$ be the last state of a minimal run of $\cA$ on $w$. Then since $\rho$ is minimal, $(w,q'_m,p)$ is a GRG (and in particular a recoverable gap), and therefore by the induction hypothesis $v_{q'_m}<\infty$. Additionally, $q'_m$ has a $\sigma$-transition to $q$, and so $v$ has a $\sigma$-transition and $\cD(w\sigma)<\infty$.

It remains to show that $\cA(w\sigma)=\cD(w\sigma)$ in the case that $\cA(w\sigma),\cD(w\sigma)<\infty$. Again let $\rho=\rho_1, \ldots ,\rho_{|w\sigma|}$ be a minimal run of $\cA$ on $w\sigma$, $q_m=\rho_{|w\sigma|}$ and $q'_m=\rho_{|w|}$. Since $\rho$ is minimal, $q'_m$ minimizes the term $\cA_{[Q_0\to q'_m]}(w)+\disc{|w|}\val(q'_m,\sigma,q_m)$ and therefore the term $\undisc{|w|}(\cA_{[Q_0\to q'_m]}(w)-\cA(w))+\val(q'_m,\sigma,q_m)$. We have already shown that $v_{q'_m}$ is finite and therefore equals $\undisc{|w|}(\cA_{[Q_0\to q'_m]}(w)-\cA(w))$. Additionally, by the induction hypothesis, $v_{q'}$ equals either $\undisc{|w|}(\cA_{[Q_0\to q']}(w)-\cA(w))$ or $\infty$ for every $q'\in Q$. Recall that $u'_{q_m}=\min_{q'\in Q}(v_{q'}+\val(q',\sigma,q_m))$. Therefore the transition $(v,\sigma,u)$ satisfies $u'_{q_m}=v_{q'_m}+\val(q'_m,\sigma,q_m)$. We now claim that $u'_{q_m}=\min_{q\in Q}u'_q$. Since $\cA_{[Q_0\to q_m]}(w\sigma)=\cA(w\sigma)$, $q_m$ minimizes the term $\min_{q'\in Q}(\cA_{[Q_0\to q']}(w)+\disc{|w|}\val(q',\sigma,q_m))$ and therefore the term $\min_{q'\in Q}(\undisc{|w|}(\cA_{[Q_0\to q']}(w)-\cA(w))+\val(q',\sigma,q_m))$. Let $q\in Q$. Recall that $u'_q=\min_{q'\in Q}(v_q'+\val(q',\sigma,q)$. since $v_{q'}$ equals either $\undisc{|w|}(\cA_{[Q_0\to q']}(w)-\cA(w))$ or $\infty$ for every $q'\in Q$, we have $u'_q\geq\min_{q'\in Q}(\undisc{|w|}(\cA_{[Q_0\to q']}(w)-\cA(w))+\val(q',\sigma,q)\geq u'_{q_m}$, and so $u'_{q_m}$ is the minimal entry of $u'$, as needed. Therefore $r=u'_{q_m}$ and so $u_{q_m}=0$ and $\val_D(v,\sigma,u)=u'_{q_m}$. Now, 
\begin{align*}
    \undisc{|w|}\cD(w\sigma)= & \undisc{|w|}\cD(w)+\val_D(v,\sigma,u)\\
    = & \undisc{|w|}\cD(w)+u'_{q_m}\\
    = & \undisc{|w|}\cD(w)+v_{q'_m}+\val(q'_m,\sigma,q_m)\\
    \overset{(*)}{=} & \undisc{|w|}\cA(w)+\undisc{|w|}(\cA_{[Q_0\to q'_m]}(w)-\cA(w))+\val(q'_m,\sigma,q_m)\\
    = & \undisc{|w|}\cA_{[Q_0\to q'_m]}(w)+\val(q'_m,\sigma,q_m)\\
    = & \undisc{|w|}\cA(w\sigma)
\end{align*}
where $(*)$ is by the induction hypothesis, and multiplying both sides by $\disc{|w|}$ gives the desired equality.

\item 
We now assume that $\cD(w\sigma)<\infty$ and show that for every $q\in Q$, the gap $\undisc{|w\sigma|}(\cA_{[Q_0\to q]}(w\sigma)-\cA(w\sigma))$ is represented correctly by $u_q$ provided that it is recoverable, where $u=\delta^*_D(v_0,w\sigma)$. Let $q\in Q$. First assume $\cA_{[Q_0\to q]}(w\sigma)=\infty$. Assume by way of contradiction that $u_q<\infty$, and let $q'$ be such that $u'_q=v_{q'}+\val(q',\sigma,q)$, where $v=\delta^*_D(v_0,w)$. Then $v_{q'}<\infty$, and by the induction hypothesis $\cA_{[Q_0\to q']}(w)<\infty$, that is, $\cA$ has a run on $w$ ending in $q'$. Also, $\val(q',\sigma,q)<\infty$, that is, $q'$ has a $\sigma$-transition. Therefore $\cA$ has a run on $w\sigma$ and we have $\cA(w\sigma)<\infty$.

Now assume $\cA_{[Q_0\to q]}(w\sigma)<\infty$. Let $\rho^q$ be a minimal run among the runs of $\cA$ on $w\sigma$ that end in $q$. Let $q'$ be the state $\rho^q$ visits after reading $w$. Since $\rho^q$ is minimal among the runs ending in $q$, $q'$ minimizes the term $\cA_{[Q_0\to q']}(w)+\disc{|w|}\val(q',\sigma,q)$ and therefore the term $\undisc{|w|}(\cA_{[Q_0\to q']}(w)-\cA(w))+\val(q',\sigma,q)$.  Consider the following cases:
\begin{itemize}
    \item $(w,q',p')$ is a recoverable gap, where $p'$ is the last state in some minimal run of $\cA$ on $w$. Then $\undisc{|w|}(\cA_{[Q_0\to q']}(w)-\cA(w))\leq B$, since it is the size of a recoverable gap. By the induction hypothesis, $v_{q'}=\undisc{|w|}(\cA_{[Q_0\to q']}(w)-\cA(w))$, and for any other $q''\in Q$ such that $v_{q''}<\infty$ it holds that $v_{q''}=\undisc{|w|}(\cA_{[Q_0\to q'']}(w)-\cA(w))$. Therefore $u'_q=\undisc{|w|}(\cA_{[Q_0\to q']}(w)-\cA(w))+\val(q',\sigma,q)$. Define $q_m,q'_m$ as before, then we have
    \begin{align*}
        u'_q-r= & \undisc{|w|}(\cA_{[Q_0\to q']}(w)-\cA(w))+\val(q',\sigma,q)- \\
        & (\undisc{|w|}(\cA_{[Q_0\to q'_m]}(w)-\cA(w))+\val(q'_m,\sigma,q_m))\\
        = & \undisc{|w|}(\cA_{[Q_0\to q']}(w)+\disc{|w|}\val(q',\sigma,q))- \\
        & \undisc{|w|}(\cA_{[Q_0\to q'_m]}(w)+\disc{|w|}\val(q'_m,\sigma,q_m))\\
        = & \undisc{|w|}(\cA_{[Q_0\to q]}(w\sigma)-\cA(w\sigma))
    \end{align*}
    Multiplying both sides by $\lam$, we get that $u_q$ equals $\undisc{|w\sigma|}(\cA_{[Q_0\to q]}(w\sigma)-\cA(w\sigma))$ if $\undisc{|w\sigma|}(\cA_{[Q_0\to q]}(w\sigma)-\cA(w\sigma))\leq B$ and $\infty$ otherwise. Also, if $(w,q,p)$ is a recoverable gap where $p$ is the last state of some minimal run of $\cA$ on $w\sigma$, then $\undisc{|w\sigma|}(\cA_{[Q_0\to q]}(w\sigma)-\cA(w\sigma))\leq B$ and so $u_q=\undisc{|w\sigma|}(\cA_{[Q_0\to q]}(w\sigma)-\cA(w\sigma))$ as needed.
    \item $(w,q',p')$ is not a recoverable gap, where $p'$ is the last state of some minimal run of $\cA$ on $w$. If $v_{q'}<\infty$ then the analysis is identical to the previous case. If $v_{q'}=\infty$ then $u_q=\infty$. It remains to show that $(w\sigma,q,p)$ is not a recoverable gap, where $p$ is the last state in some minimal run of $\cA$ on $w\sigma$. Assume by way of contradiction that $(w\sigma,q,p)$ is a recoverable gap with respect to $z$. Then $(w,q',p')$ is a recoverable gap (of the same type) with respect to $\sigma z$, in contradiction to it not being a recoverable gap.
\end{itemize}
\end{enumerate}

Recall that our goal to show that $\cA^*(w)=\cD^*(w)$ for every $w\in\Sigma^*$. Next we show that if $\cD^*(w)=\infty$ then $\cA^*(w)=\infty$. We show the contra-positive: Assume $\cA^*(w)<\infty$, that is, $\cA$ has an accepting run on $w$. Let $q\in\alpha$ be such that $\cA_{[Q_0\to_f q]}(w)=\cA^*(w)$. In particular, $\cA$ has a run on $w$, and by the claim above so does $\cD$. Let $u=\delta_D^*(v_0,w)$. We need to show that $u\in\alpha_D$, and for that it is enough to show that $u_q<\infty$. This is implied by the claim above and the fact that $\cA_{[Q_0\to q]}(w)<\infty$, provided that $(w,q,p)$ is a recoverable gap, where $p$ is the last state in some minimal run of $\cA$ on $w\sigma$. Since $\cA_{[Q_0\to_f q]}(w)=\cA^*(w)$, that is obviously the case: $(w,q,p)$ is a TRG with respect to $z=\epsilon$.

Finally, we show that if $\cD^*(w)<\infty$ then $\cA^*(w)=\cD^*(w)$. Let $u=\delta^*_D(v_0,w)$. Then $\cD^*(w)=\cD(w)+\disc{|w|}\fval_D(u)$. By the claim above we have $\cA(w)=\cD(w)$, so it remains to show that $\cA^*(w)=\cA(w)+\disc{|w|}\fval_D(u)$. It holds that 
$\cA^*(w)=\min_{q\in \alpha}(\cA_{[Q_0\to q]}(w)+\disc{|w|}\fval(q))=\cA(w)+\min_{q\in \alpha}(\cA_{[Q_0\to q]}(w)-\cA(w)+\disc{|w|}\fval(q))$, 
and therefore it remains to show that 
\[\fval_D(u)=\min_{q\in \alpha}(\undisc{|w|}(\cA_{[Q_0\to q]}(w)-\cA(w))+\fval(q)).\]
By definition, $\fval_D(u)=\min_{q\in \alpha}(u_q+\fval(q))$. By the claim above, $(\cA_{[Q_0\to q]}(w)-\cA(w))=u_q$ for all those $q$ that have $u_q<\infty$, and so it remains to show that any $q$ that minimizes the term $\undisc{|w|}(\cA_{[Q_0\to q]}(w)-\cA(w))+\fval(q)$ satisfies $u_q<\infty$. Assume by way of contradiction that that is not the case. By the claim above we have that the gap $(w,q,p)$, where $p$ is the last state in some minimal run of $\cA$ on $w$, is unrecoverable. But $q$ minimizes the term $\cA_{[Q_0\to q]}(w)+\disc{|w|}\fval(q)$, meaning $\cA_{[Q_0\to_f q]}(w)=\cA^*(w)$, so the gap is a TRG with respect to $z=\epsilon$.
\hfill\qed

\subsection{Proof of~\cref{lem: a large gap has to increase}}
\label{apx:a large gap has to increase}
We have $\Gamma(\rho_0^1, \ldots ,\rho_{n+1}^1) \geq \lam(\Gamma(\rho_0^1, \ldots ,\rho_n^1)-m_\cA)$, since in the worst case the last transition weighs $-m_{\cA}$. Similarly, $\Gamma(\rho_0^2, \ldots ,\rho_{n+1}^2) \leq \lam(\Gamma(\rho_0^2, \ldots ,\rho_n^2)+m_\cA)$. It follows that $\Gamma(\rho_0^1, \ldots ,\rho_{n+1}^1)-\Gamma(\rho_0^2, \ldots ,\rho_{n+1}^2)\geq \lam(\Gamma(\rho_0^1, \ldots ,\rho_n^1)-\Gamma(\rho_0^2, \ldots ,\rho_n^2)-2m_{\cA})$. Now, we have:

\begin{align*}
\Gamma(\rho_0^1, \ldots ,\rho_n^1)-\Gamma(\rho_0^2, \ldots ,\rho_n^2)> & 2\frac{\lambda}{\lambda-1}m_\cA\\
(\lambda -1)(\Gamma(\rho_0^1, \ldots ,\rho_n^1)-\Gamma(\rho_0^2, \ldots ,\rho_n^2))> & 2\lambda m_\cA\\
\lambda (\Gamma(\rho_0^1, \ldots ,\rho_n^1)-\Gamma(\rho_0^2, \ldots ,\rho_n^2))-2\lambda m_\cA> & \Gamma(\rho_0^1, \ldots ,\rho_n^1)-\Gamma(\rho_0^2, \ldots ,\rho_n^2)\\
\lambda (\Gamma(\rho_0^1, \ldots ,\rho_n^1)-\Gamma(\rho_0^2, \ldots ,\rho_n^2)-2m_\cA)> & \Gamma(\rho_0^1, \ldots ,\rho_n^1)-\Gamma(\rho_0^2, \ldots ,\rho_n^2) \quad\qed\\
\end{align*}
\hfill\qed

\subsection{Proof of~\cref{lem: it takes time to reach a large gap}}
\label{apx: it takes time to reach a large gap}
Recall that $\cM=2\frac{\lam}{1-\lam}m_\cA$ and Let $N=\undisc{n_\text{steps}}(n_\text{gap}-\cM)+\cM$. We show that $N$ satisfies the requirement. Let $f(x)=\lam(x+2m_\cA)$. For every $0\leq i<n$, we have $\Gamma(\rho^1_0, \ldots ,\rho^1_{i+1})\leq \lam(\Gamma(\rho^1_0, \ldots ,\rho^1_i)+m_\cA)$, since in the worst case the last transition weighs $m_\cA$. Similarly, $\Gamma(\rho^2_0, \ldots ,\rho^2_{i+1})\geq \lam(\Gamma(\rho^2_0, \ldots ,\rho^2_i)-m_\cA)$, and so
\begin{align*}
    \Gamma(\rho^1_0, \ldots ,\rho^1_{i+1})-\Gamma(\rho^2_0, \ldots ,\rho^2_{i+1})\leq & \lam(\Gamma(\rho^1_0, \ldots ,\rho^1_i)-\Gamma(\rho^2_0, \ldots ,\rho^2_i)+2m_\cA) \\
    = & f(\Gamma(\rho^1_0, \ldots ,\rho^1_i)-\Gamma(\rho^2_0, \ldots ,\rho^2_i))
\end{align*}
Therefore, if $\Gamma(\rho^1_0, \ldots ,\rho^1_i)-\Gamma(\rho^2_0, \ldots ,\rho^2_i)\leq n_\text{gap}$ then $\Gamma(\rho^1_0, \ldots ,\rho^1_n)-\Gamma(\rho^2_0, \ldots ,\rho^2_n)\leq f^{n-i}(n_\text{gap})$. Thus, the statement of the lemma holds provided that $N=f^{n_\text{steps}}(n_\text{gap})$. 

In order to show this equality, we first observe that $f(x)=\lam(x-\cM)+\cM$. Indeed, 
\[
\lam\cdot2m_\cA=  (1-\lam)\frac{\lam\cdot2m_\cA}{1-\lam}
=  -\lam\frac{\lam\cdot2m_\cA}{1-\lam}+\frac{\lam\cdot2m_\cA}{1-\lam}\]

and therefore
\begin{align*}
f(x)= & \lam(x+2m_\cA) = \lam x-\lam\frac{\lam\cdot2m_\cA}{1-\lam}+\frac{\lam\cdot2m_\cA}{1-\lam}\\
= & \lam(x-\frac{\lam\cdot2m_\cA}{1-\lam})+\frac{\lam\cdot2m_\cA}{1-\lam} = \lam(x-\cM)+\cM
\end{align*}

It now easily follows that $f^k(x)=\lam^k(x-\cM)+\cM$ for every $k\in \bbN$. For example:
\begin{align*}
f^2(x)= & f(\lam(x-\cM)+\cM) = \lam((\lam(x-\cM)+\cM)-\cM)+\cM\\
= & \lam\cdot\lam(x-\cM)+\cM = \lam^2(x-\cM)+\cM
\end{align*}

Thus, we have
$f^{n_\text{steps}}(n_\text{gap})=\undisc{n_\text{steps}}(n_\text{gap}-\cM)+\cM=N$, as required.
\hfill\qed

\subsection{Proof of~\cref{lem: Udi's constant unrecoverablility - with final weights}}
\label{apx: Udi's constant unrecoverablility - with final weights}
\begin{align*}
    |\val(\rho^{u_f})+\disc{|z|}\fval(q_{u_f})| \leq & \Sigma_{i=0}^{|z|-1}\disc{i}|\val(\rho^{u_f}_{i},z_{i+1},\rho^{u_f}_{i+1})|+\disc{|z|}|\fval(q_{u_f})| \\
    \leq & \Sigma_{i=0}^{|z|}\disc{i}m_\cA 
    < \frac{\lam}{\lam-1}m_\cA
\end{align*}
And similarly, $|\val(\rho^{l_f})+\disc{|z|}\fval(q_{l_f})|<\frac{\lam}{\lam-1}m_\cA$. Therefore, the difference between the weights is lower than $2\frac{\lam}{\lam-1}m_\cA=\cM$ in absolute value. Now,
\begin{align*}
    & \val(\rho^u\rho^{u_f})+\disc{|wz|}\fval(q_{u_f})-(\val(\rho^l\rho^{l_f})+\disc{|wz|}\fval(q_{l_f})) \\
    > & \val(\rho^u)-\val(\rho^l) - \disc{|w|}\cM > \disc{|w|}\cM - \disc{|w|}\cM = 0
\end{align*}
\hfill\qed

\subsection{Proof of~\cref{lem: short w}}
\label{apx: short w}
Assume that $w$ has the $\cN$-separation property
with respect to\emph{ $(U,L,q_u,z)$}.

We use an identical construction to that of \cref{lem: finite gaps -> det}, with bound $C$, in order to define a sequence of vectors $v_0,\ldots,v_{|w|}$ with $v_i\in \{0,\ldots,C,\infty\}^Q$ for every $0\le i\le |w|$ that, intuitively, keep track of the runs of $\cA$ on $w$, as follows. 
\begin{itemize}
\item For every $q\in Q$ set $(v_{0})_q=\begin{cases}
    0 & q\in Q_0\\
    \infty & \text{otherwise}
\end{cases}$
\item For every $i>0,q\in Q$ let $v_{i,q}'=\min_{q'\in Q}((v_{i-1})_{q'}+\val(q',w_i,q))$, where $\val(q',\sigma,q)$ is regarded as $\infty$ if $(q',\sigma,q)\notin \delta$ (the $v'_{i,q}$ are ``intermediate'' values).
\item For every $i>0$ let $r_i=\min_{q\in Q}v_{i,q}'$ (the $r_i$ are the offset of the vector from $0$). 
\item For every $i>0,q\in Q$ set $(v_{i})_q=\begin{cases}
    \lambda(v_{i,q}'-r_i) & \lambda(v_{i,q}'-r_i)\leq C\\
    \infty & \text{otherwise}
\end{cases}$
\end{itemize}

Recall that intuitively, $(v_i)$ tracks, for each $q\in Q$, the gap between the minimal run on $w_1 \cdots w_i$ ending in $q$ and the minimal run on this prefix overall. When this gap becomes large enough that recovering from it implies the existence of $\cN$-separation, it is denoted $\infty$.

Denote the normalized difference $\undisc{i} (\cA_{[Q_0\to q]}(w_{1} \cdots w_{i})-\cA(w_{1} \cdots w_{i}))$ by $\Delta_{q,i}(w)$. It is easy to show that $v_i$ keeps the correct weight of runs whose gap from the minimal one remains always under $C$. However, if a gap of a run goes over $C$ but then comes back down, then $v_i$ no longer tracks it correctly. To account for this, we claim that since $w$ has the $\cN$-separation property, for every $q,i$ at least one of the following must
hold:
% \begin{itemize}
% \item $(v_{i})_q=\begin{cases}
% \undisc{i} (\cA_{[Q_0\to q]}(w_{1} \cdots w_{i})-\cA(w_{1} \cdots w_{i})) & \undisc i (\cA_{[Q_0\to q]}(w_{1} \cdots w_{i})-\cA(w_{1} \cdots w_{i}))\leq C\\
% \infty & \text{otherwise}
% \end{cases}$.
% \ndtodo{I think the only way to fix this width (which repeats later) is to give some short notation to that expression, does that make sense?}
% \item There exists $i'<i$ such that $w_{1} \cdots w_{i'}$ has the $\cN$-separation property.
% \end{itemize}
\begin{itemize}
\item $(v_{i})_q=\begin{cases}
\Delta_{q,i}(w) & \Delta_{q,i}(w)\leq C\\
\infty & \text{otherwise}
\end{cases}$.
\item There exists $i'<i$ such that $w_{1} \cdots w_{i'}$ has the $\cN$-separation property.
\end{itemize}
That is, either $v_i$ tracks the runs correctly, or there is some shorter prefix that already has the $\cN$-separation property.

The proof is by induction on $i$:

The case $i=0$ is trivial. Let $i>0$. If there exists $i'<i-1$ such that $w_{1} \cdots w_{i'}$ has the $\cN$-separation property, we are done. Otherwise, let $q,q'\in Q$ be states such that $v_{i,q}'=((v_{i-1})_{q'}+\val(q',w_i,q))$.

Intuitively, we now claim that the only problematic case arises when $(v_{i-1})_{q'}=\infty$, and so the information about the exact value of the gap represented by $(v_{i-1})_{q'}$ is gone. We consider the normalization value $r_i$ (i.e., the offset of the minimal run from $0$): if $r_i$ is small, then the gap represented by $(v_i)_q$ is still very large, and we show that marking it as $\infty$ is sound. Otherwise, if $r_i$ is large, then the above gap might indeed be wrongly marked as $\infty$. However, we show that in this case, $r_i$ is so large that we can actually obtain an $\cN$-separation property ``below'' $r_i$, using a shorter witness.

Formally, consider the following cases:
\begin{itemize}
\item $(v_{i-1})_{q'}<\infty$. Let $q_m$ be a state such that $v_{i,q_m}'$ is minimal (and therefore equals $r_i$). Then $q_m$ minimizes the term $\min_{q''\in Q}((v_{i-1})_{q''}+\val(q'',w_i,q_m))$, which by the induction hypothesis equals either $\min_{q''\in Q}(\undisc{i-1}(\cA_{[Q_0\to q'']}(w_1 \cdots w_{i-1})-\cA(w_1 \cdots w_{i-1}))+\val(q'',w_i,q_m))$ or $\infty$, and therefore minimizes the term $\min_{q''\in Q}(\cA_{[Q_0\to q'']}(w_1 \cdots w_{i-1})+\undisc{i-1}\val(q'',w_i,q_m))=\cA_{[Q_0\to q_m]}(w_1 \cdots w_i)$, meaning $\cA_{[Q_0\to q_m]}(w_1 \cdots w_i)=\cA(w_1 \cdots w_i)$. Let $q_m'$ be a state such that $v_{i,q_m}'=((v_{i-1})_{q_m'}+\val(q_m',w_i,q_m))$. We have:
\begin{align*}
v_{i,q}'-r_i = & (v_{i-1})_{q'}+\val(q',w_i,q)-(v_{i-1})_{q_m'}-\val(q_m',w_i,q_m)\\
\overset{(1)}{=} & \undisc{i-1}(\cA_{[Q_0\to q']}(w_1 \cdots w_{i-1})-\cA(w_1 \cdots w_{i-1})+\val(q',w_i,q)-\\
  & \undisc{i-1}(\cA_{[Q_0\to q_m']}(w_1 \cdots w_{i-1})-\cA(w_1 \cdots w_{i-1})+\val(q_m',w_i,q_m))\\
= & \undisc{i-1}(\cA_{[Q_0\to q']}(w_1 \cdots w_{i-1})+\disc{(i-1)}\val(q',w_i,q))-\\
  & \undisc{i-1}(\cA_{[Q_0\to q_m']}(w_1 \cdots w_{i-1})+\disc{(i-1)}\val(q_m',w_i,q_m))\\
\overset{(2)}{=} & \undisc{i-1}(\cA_{[Q_0\to q]}(w_1 \cdots w_i)-\cA(w_1 \cdots w_i))
\end{align*}
where (1) is due to the induction hypothesis. As for (2), the choice of $q'$ among all states in $Q$ minimizes the term $(v_{i-1})_{q'}+\val(q',w_i,q)$, which again by the induction hypothesis equals $\undisc{i-1}(\cA_{[Q_0\to q']}(w_1 \cdots w_{i-1})-\cA(w_1 \cdots w_{i-1})+\disc{(i-1)}\val(q',w_i,q))$, and so it also minimizes the term $\cA_{[Q_0\to q']}(w_1 \cdots w_{i-1})+\disc{(i-1)}\val(q',w_i,q)$. It follows that $\cA_{[Q_0\to q']}(w_1 \cdots w_{i-1})+\disc{(i-1)}\val(q',w_i,q)=\cA_{[Q_0\to q]}(w_1 \cdots w_i)$. Similarly, $\cA_{[Q_0\to q_m']}(w_1 \cdots w_{i-1})+\disc{(i-1)}\val(q_m',w_i,q_m)=\cA_{[Q_0\to q_m]}(w_1 \cdots w_i)=\cA(w_1 \cdots w_i)$, and (2) follows. Finally, the desired equality follow by multiplying both sides by $\lam$, whether $\undisc i (\cA_{[Q_0\to q]}(w_{1} \cdots w_{i})-\cA(w_{1} \cdots w_{i}))\leq C$ or not.
\item $(v_{i-1})_{q'}=\infty$ and $r_i\leq C\frac{\lambda-1}{\lambda}-m_{\cA}$.
Since $(v_{i-1})_{q'}=\infty$, we have $(v_i)_q=\infty$.
It remains to show that $\undisc i(\cA_{[Q_0\to q]}(w_{1} \cdots w_{i})-\cA(w_{1} \cdots w_{i}))>C$.
Indeed, 
\begin{align*}
& \undisc i(\cA_{[Q_0\to q]}(w_{1} \cdots w_{i})-\cA(w_{1} \cdots w_{i})) \\
\geq & \undisc{i}(\cA_{[Q_0\to q']}(w_{1} \cdots w_{i-1})-\cA(w_{1} \cdots w_{i-1})-(m_{\cA}+r_i)\cdot\disc{(i-1)})\\
= & \lam(\undisc{i-1}(\cA_{[Q_0\to q']}(w_1 \cdots w_{i-1})-\cA(w_1 \cdots w_{i-1}))-r_i-m_{\cA})\\
> & \lam(C-(C\frac{\lam-1}{\lam}-m_{\cA})-m_{\cA}) = \lam(\lami C+m_{\cA}-m_{\cA}) > C
\end{align*}
where the first transition follows from the fact that when reading $w_i$, in the worst case, the weight of a specific run can decrease by $\disc{(i-1)}m_{\cA}$, and the overall weight of the word can increase by $\disc{(i-1)}r_i$; indeed, let $q_m$ be a state such that $v_{i,q_m}'$ is minimal (and therefore equals $r_i$), and let $q_m'$ be a state such that $v_{i,q_m}'=((v_{i-1})_{q_m'}+\val(q_m',w_i,q_m))$. Then

\begin{align*}
\cA(w_1 \cdots w_i) \leq & \cA_{[Q_0\to q_m]}(w_1 \cdots w_i) \\
\leq & \cA_{[Q_0\to q'_m]}(w_1 \cdots w_{i-1}) + \disc{i}\val(q_m',w_i,q_m) \\
\overset{(*)}{=} & \cA(w_1 \cdots w_{i-1})+\disc{(i-1)}(v_{i-1})_{q'_m}+\disc{i}\val(q_m',w_i,q_m) \\
< & \cA(w_1 \cdots w_{i-1})+\disc{(i-1)}((v_{i-1})_{q'_m}+\val(q_m',w_i,q_m)) \\
= & \cA(w_1 \cdots w_{i-1})+\disc{(i-1)}r_1
\end{align*}
Where $(*)$ is due to the induction hypothesis.

\item $r_i>C\frac{\lambda-1}{\lambda}-m_{\cA}$. This is only possible
if for every $q_l$ such that $(v_{i-1})_{q_l}<C\frac{\lambda-1}{\lambda}-2m_{\cA}=\cN|Q|$,
$q_{l}$ has no $w_{i}$-transition. Let $L''=\{ q_{l}\in Q\mid(v_{i-1})_{q_l}<\cN|Q|\} $.
Write $Q={q_1, \ldots ,q_{|Q|}}$
such that $(v_{i-1})_{q_1}\leq \ldots \leq(v_{i-1})_{q_{|Q|}}$, and so $L''=\{q_1, \ldots ,q_{|L''|}\}$. Since $w$ has the $\cN$-separation property, in particular $\cA$ has a run on $w$ and so $L''\subsetneq Q$. Then, there exists $1\leq r\leq |L''|$ such that $(v_{i-1})_{q_{r+1}}-(v_{i-1})_{q_r}>\cN$.
Let $U'=\{q_{r+1}, \ldots ,q_{|Q|}\},L'=\{q_1, \ldots ,q_r\}$,
and note that for every $q_l'\in L'$, $q_l'$
has no $w_{i}$-transition. For every $q_l'\in L',q_u'\in U'$, we have $\undisc{i-1}(\cA_{[Q_0\to q_u']}(w_1 \cdots w_{i-1})-\cA_{[Q_0\to q_l']}(w_1 \cdots w_{i-1})=(v_{i-1})_{q_u'}-(v_{i-1})_{q_l'}>\cN$.  Let $q_u'\in U'$ be such that $\cA_{[Q_0\to q_u']}(w_1 \cdots w_{i-1})+\disc{(i-1)}\cA_{[q_u'\to_f \alpha]}(w_i)$ is minimal. Then for every $q_l'\in L'$, $(w_1 \cdots w_{i-1},q_u',q_l')$ is a TRG with respect to $w_i$, and so $w_1 \cdots w_{i-1}$ has the $\cN$-separation property with respect to $(U',L',q_u',w_i)$, and we are done.
\end{itemize}

Now, it remains to show that if $|w|>(C+2)^{|Q|}$,
there exists $w'$ such that $|w'|<|w|$
and $w'$ has the $\cN$-separation property.
If $w_{1} \cdots w_{i'}$ has the $\cN$-separation
property for some $i'<|w|$, we are done. 
Otherwise, for every $q,i$ we have 

\[
(v_{i})_q=\begin{cases}
\Delta_{q,i}(w) & \Delta_{q,i}(w)\leq C\\
\infty & \text{otherwise}
\end{cases}
\]

In particular, for every $i$, $\min_{q\in Q}(v_i)_q=0$.
Therefore there can only be $(C+2)^{|Q|}$ different vectors among $(v_{i})_{i=0}^{|w|}$, and so by the pigeonhole principle there exist $0\leq i_{1}<i_{2}\leq|w|$ with $v_{i_1}=v_{i_2}$.
We remove the infix between $i_1$ and $i_2$ and consider $w'=w_{1} \cdots w_{i_{1}}w_{i_{2}+1} \cdots w_{|w|}$.
We can construct vectors $\bar{v}_{0}, \ldots ,\bar{v}_{|w'|}$
for the word $w'$ similarly to the construction above, and a
similar claim about the vectors holds. Since $v_{i_{1}}=v_{i_{2}}$,
we have $\bar{v}_{|w'|}=v_{|w|}$. If $w'_1 \cdots w'_{i'}$
has the $\cN$-separation property for some $i'<|w'|$,
then obviously $i'<|w|$ and so we are done. Otherwise, for every $q,i$ we have
\[
(\bar{v}_{i})_q=\begin{cases}
\Delta_{q,i}(w') & \Delta_{q,i}(w')\leq C\\
\infty & \text{otherwise}
\end{cases}
\]

We can assume, without loss of generality, that $(\bar{v}_{|w'|})_{q_l}=(v_{|w|})_{q_l}\leq \cN(|Q|-1)$
for every $q_l\in L$. Indeed, if that is not the case, then ordering the states according to $(v_{|w|})_q$, there are two consecutive states $q_1,q_2\in L$ such that $(v_{|w|})_{q_1}-(v_{|w|})_{q_2} > \cN$. Let $L'$ be the set of states lower than $q_1$ and $U'=Q\backslash L'$. $L'$ is not empty as $q_2\in L$. Additionally, $L'\subseteq L$ and
$q_u\in U'$, and so $w$ has the $\cN$-separation
property with respect to $(U',L',q_u)$ as
well.

It remains to show that (i) $\undisc{|w'|}(\cA_{[Q_0\to q_u']}(w')-\cA_{[Q_0\to q_l]}(w'))>\cN$ for every $q_l\in L,q_u'\in U$; and (ii) There exists $q_u'\in U$ such that $(w',q_u',q_l')$ is a TRG with respect to $z$ for every $q_l\in L$.

For (i), let $q_l\in L,q_u'\in U$. If $(\bar{v}_{|w'|})_{q_u'}<\infty$, then
\begin{align*}
& \undisc{|w'|}(\cA_{[Q_0\to q_u']}(w')-\cA_{[Q_0\to q_l]}(w')) = (\bar{v}_{|w'|})_{q_u'}-(\bar{v}_{|w'|})_{q_l}\\
= & (v_{|w|})_{q_u'}-(v_{|w|})_{q_l} = \undisc{|w|}(\cA_{[Q_0\to q_u']}(w)-\cA_{[Q_0\to q_l]}(w)) > \cN
\end{align*}

If $(\bar{v}_{|w'|})_{q_u'}=\infty$, then $\undisc{|w'|}(\cA_{[Q_0\to q_u']}(w')-\cA(w'))>C$ and so
\begin{align*}
& \undisc{|w'|}(\cA_{[Q_0\to q_u']}(w')-\cA_{[Q_0\to q_l]}(w')) \\
= & \undisc{|w'|}(\cA_{[Q_0\to q_u']}(w')-\cA(w'))-\undisc{|w'|}(\cA_{[Q_0\to q_l']}(w')-\cA(w'))\\
= & \undisc{|w'|}(\cA_{[Q_0\to q_u']}(w')-\cA(w'))-(\bar{v}_{|w'|})_{q_l}\\
> & C-\cN(|Q|-1) > \cN|Q|-\cN(|Q|-1) = \cN
\end{align*}

For (ii), let $q_u'\in U$ be such that $\cA_{[Q_0\to q_u']}(w')+\cA_{[q_u'\to_f \alpha]}(z)$ is minimal. Then $q_u'$ satisfies the requirement.

In conclusion, $w'$
has the $\cN$-separation property with respect to $(U,L,q_u')$, as needed.
\hfill\qed 

\subsection{Proof of~\cref{lem: short z}}
\label{apx:short z}
By \cref{lem: Udi's constant unrecoverablility - gaps} and that fact that $\cN>\cM$, it is enough to show that if for two sets of states $U,L$ and a word $z$ it holds that $\cA_{[U\to_f \alpha]}(z)<\infty$ and $\cA_{[L\to_f \alpha]}(z)=\infty$, then the same is true for some $z'$ such that $|z'|\leq 2^{2|Q|}$. Consider the NFAs $\cA_u,\cA_l$ whose states and transitions are identical to those of $\cA$, with $U,L$ as sets of initial states, respectively. A word $z'$ satisfies $\cA_{[U\to_f \alpha]}(z')<\infty$
and $\cA_{[L\to_f \alpha]}(z')=\infty$ if and only if $z'\in L(\cA_U)\backslash L(\cA_L)$. The claim is therefore true due to the fact that for two NFAs $\cA_1,\cA_2$ with sets of states $Q_1,Q_2$ and languages $\mathcal{L}_1,\mathcal{L}_2$, there exists $z'\in \mathcal{L}_1\backslash\mathcal{L}_2$ if and only if there exists such $z'$ of length at most $2^{|Q_1|+|Q_2|}$.
\hfill\qed

\end{document}